\documentclass[11pt]{article}
\usepackage[margin=1in]{geometry}
\usepackage[utf8]{inputenc} %
\usepackage{graphicx,xcolor,xspace} %
\usepackage{amsthm,amsmath,amssymb,mathtools}
\usepackage{tcolorbox}
\usepackage{subfig}
\usepackage{enumerate}
\tcbuselibrary{skins,breakable}
\tcbset{enhanced jigsaw}

\usepackage{thmtools} 
\usepackage{thm-restate}

\usepackage[linktocpage=true,
	pagebackref=true,colorlinks,citecolor=olive,
	bookmarks,bookmarksopen,bookmarksnumbered]
	{hyperref}
\usepackage[noabbrev,nameinlink]{cleveref}

\usepackage{todonotes}

\newtheorem{theorem}{Theorem}
\newtheorem{lemma}{Lemma}[section]
\newtheorem{definition}[lemma]{Definition}
\newtheorem{proposition}[lemma]{Proposition}
\newtheorem{observation}[lemma]{Observation}

\newtheorem{fact}[lemma]{Fact}

\newtheorem{problem}[lemma]{Problem}

\crefname{lemma}{Lemma}{Lemmas}
\crefname{appendix}{Appendix}{Appendices}
\crefname{proposition}{Proposition}{Propositions}
\crefname{observation}{Observation}{Observations}
\crefname{claim}{Claim}{Claims}
\crefname{figure}{Figure}{Figures}

\title{Faster Approximation Algorithms for Restricted Shortest Paths\\in Directed Graphs}
\author{Vikrant Ashvinkumar \thanks{Rutgers.}
\and Aaron Bernstein \thanks{Supported by Sloan Fellowship, Google Research Fellowship, NSF Grant 1942010, and Charles S. Baylis endowment at NYU.}
\and Adam Karczmarz \thanks{University of Warsaw and IDEAS NCBR, Poland. Partially supported by the ERC
CoG grant TUgbOAT no 772346 and the National Science Centre (NCN) grant no. 2022/47/D/ST6/02184.}}
\date{}
\newcommand{\eps}{\ensuremath{\varepsilon}}
\newcommand{\Ot}{\ensuremath{\widetilde{O}}}
\newcommand{\bracket}[1]{\left[#1\right]}
\newcommand{\paren}[1]{\left(#1\right)}
\newcommand{\card}[1]{\left\vert{#1}\right\vert}
\newcommand{\set}[1]{\left\{ #1 \right\}}
\newcommand{\ceil}[1]{{\left\lceil{#1}\right\rceil}}
\newcommand{\floor}[1]{{\left\lfloor{#1}\right\rfloor}}
\newcommand{\IR}{\ensuremath{\mathbb{R}}}

\newcommand{\IZ}{\ensuremath{\mathbb{Z}}\xspace}
\newcommand{\IN}{\ensuremath{\mathbb{N}}}
\newcommand{\poly}{\mathrm{poly}}

\newcommand{\prob}[1]{\Pr\bracket{#1}}
\newcommand{\expect}[1]{\Exp\bracket{#1}}

\DeclareMathOperator*{\Exp}{\ensuremath{{\mathbb{E}}}}
\newcommand{\algname}[1]{\underline{\textbf{\textsc{#1}}}\\}
\newenvironment{tbox}{\begin{tcolorbox}[
		enlarge top by=5pt,
		enlarge bottom by=5pt,
		 breakable,
		 boxsep=2pt,
                  left=5pt,
                  right=7pt,
                  top=10pt,
                  arc=0pt,
                  boxrule=1pt,toprule=1pt,
                  colback=white
                  ]%
	}
{\end{tcolorbox}}
\newcommand{\cA}{\mathcal{A}}
\newcommand{\cC}{\mathcal{C}}
\newcommand{\hatell}{\hat{\ell}}
\newcommand{\hatd}{\hat{d}}
\newcommand{\Gin}{G_{\textrm{in}}}
\newcommand{\lendel}[1]{\mathrm{comb}(#1)}
\newcommand{\eintra}{\mathtt{IntraBlock}}
\newcommand{\eforward}[1]{\mathtt{Forward}_{#1}}
\newcommand{\eback}[1]{\mathtt{Back}_{#1}}
\newcommand{\estar}[1]{\mathtt{Star}_{#1}}
\newcommand{\ehop}{\mathtt{Hop}}
\newcommand{\edead}{\mathtt{Dead}}
\newcommand{\dist}{\mathrm{dist}}
\newcommand{\rdist}{\dist}
\newcommand{\bump}{d_{\uparrow}}

\newcommand{\pisum}{\pi_{\textrm{sum}}}
\newcommand{\expround}[2]{\mathrm{expround}_{#1}\paren{#2}}
\newcommand{\dptable}[2]{\mathsf{DP}\paren{#1, #2}}
\newcommand{\pidp}[1]{#1\textrm{-}\textsc{DP}}

\newcommand{\LDD}{\textsc{Compute}\textrm{-}\textsc{LDD}}
\newcommand{\SCC}{\textsc{Compute}\textrm{-}\textsc{SCCs}}
\newcommand{\one}{\mathbf{1}}
\newcommand{\rsp}{\mathsf{Directed}\textrm{-}\mathsf{RSP}}
\newcommand{\grsp}{\mathsf{RSP}}

\newcommand{\apx}{O(\eps \log n)}
\newcommand{\hapx}{O(\eps \log n)}

\newcommand{\blocksize}{\Delta_{\textrm{block}}}
\newcommand{\blockrun}{\Delta_{\textrm{run}}}
\newcommand{\PFP}{\hyperref[def:good-edges]{Pareto Frontier Preserving}\xspace}
\newcommand{\PSC}{\hyperref[def:path-sum-constraints]{Path-Sum Constraint}\xspace}
\newcommand{\PSCs}{\hyperref[def:path-sum-constraints]{Path-Sum Constraints}\xspace}

\begin{document}
\maketitle
\begin{abstract}
    In the restricted shortest paths problem, we are given a graph $G$ whose edges are assigned two non-negative weights: lengths and delays, a source $s$, and a delay threshold $D$.
    The goal is to find, for each target $t$, the length of the shortest $(s,t)$-path whose total delay is at most~$D$.
    While this problem is known to be NP-hard~\cite{GareyJ79}, $(1+\eps)$-approximate algorithms running in $\Ot(mn)$ time\footnote{Throughout this paper, we use the standard $\Ot(\cdot)$ notation to suppress factors polylogarithmic in $n$.}~\cite{goel2001efficient,lorenz2001simple} given more than twenty years ago have remained the state-of-the-art for \emph{directed} graphs.
    An open problem posed by \cite{Bernstein12} --- who gave a randomized $m\cdot n^{o(1)}$ time bicriteria $(1+\eps, 1+\eps)$-approximation algorithm for undirected graphs  --- asks if there is similarly an $o(mn)$ time approximation scheme for directed graphs.

    We show two randomized bicriteria $(1+\eps, 1+\eps)$-approximation algorithms that give an affirmative answer to the problem: one suited to dense graphs, and the other that works better for sparse graphs.
    On directed graphs with a quasi-polynomial weights aspect ratio\footnote{For arbitrary aspect ratios $W$, our single-pair all-targets algorithms suffer a $\log(nW)$ factor overhead in the running time.}, our algorithms run in time $\Ot(n^2)$ and, $\Ot(mn^{3/5})$ or better, respectively.
    More specifically, the algorithm for sparse digraphs runs in time $\Ot(mn^{(3 - \alpha)/5})$ for graphs with $n^{1 + \alpha}$ edges for any real $\alpha \in [0,1/2]$.
\end{abstract}

\pagenumbering{roman}
\tableofcontents
\clearpage
\pagenumbering{arabic}
\setcounter{page}{1}

\section{Introduction}
Computing shortest paths in a graph is a fundamental problem in graph algorithms. In applications, however, the paths we look for may have multiple objectives or be subject to additional constraints. For example, when planning a long journey, there are many factors to take into account: the total price, number of transfers, trip length, carbon footprint, etc. An optimal path according to one metric may be completely unacceptable in terms of another. In such scenarios, one usually optimizes one objective assuming fixed budget constraints on the other, or explores trade-offs between different factors. Problems of this kind are collectively known as \emph{multi-constrained path} problems (see, e.g.~\cite{GarroppoGT10, XueZTT08}). 

The \emph{restricted shortest paths problem} ($\grsp$) captures the problem of shortest paths with exactly one objective and one constraint; it is one of the most natural and well-studied constrained shortest path problems, with wide theoretical and real-world applications, e.g., in quality of service (QoS) routing~\cite{lorenz1998qos, younis2003constraint}, scheduling~\cite{BriskornCLLP10, NaorST07}, and other areas of operations research (see, e.g.,~\cite{GarroppoGT10}).
In the $\grsp$ problem, we are given a graph $G=(V,E)$ accompanied with \emph{exactly two} non-negative functions $\ell,d$ assigning two unrelated weights -- called \emph{lengths}
and \emph{delays} in this paper --  to the edges of~$G$. 
For a source-target pair $s,t\in V$, and a delay threshold $D$, the goal is to find the shortest (wrt. lengths $\ell$) $(s,t)$-path in $G$ whose total delay (i.e., the sum of individual edges' delays) is no more than $D$.
Just as in the case of standard shortest paths, one can also consider the single-source and all-pairs variant of the $\grsp$ problem.

\paragraph{Complexity of exact $\grsp$.} Contrary to the unconstrained problem solvable in near-linear time, the restricted shortest paths problem is NP-hard in full generality~\cite{GareyJ79}.
As a result, one can hope to solve $\grsp$ efficiently only in special cases or by allowing approximate solutions.

Given \emph{non-negative integer} lengths and delays, a pseudopolynomial $O(Dm)$ time algorithm was shown in~\cite{joksch1966shortest} for $\grsp$.
More recently, it was shown in~\cite{abboud2022seth} that even this is somewhat unlikely to be improved much further; there are no $O(D^{1-\delta}m)$ or $O(Dm^{1-\delta})$ time algorithms unless the Strong Exponential Time Hypothesis is false.

An even more special case of $\grsp$ is the \emph{hop-bounded shortest path problem}, where the edge count of the sought path is constrained by $D$, or, equivalently, the delays are all equal to $1$.
A~straightforward dynamic programming solution, that can be seen as a truncated variant of the Bellman-Ford algorithm, solves the hop-bounded shortest path problem in $O(mD)=O(mn)$ time.
\cite{Kociumaka023} proved that this simple algorithm is conditionally optimal (under the so-called Min-Plus Convolution Hypothesis).
The lower bound holds for all densities and thresholds $D$, and even if~$G$ is undirected with lengths $\ell$ bounded polynomially in $n$ and integral.
Therefore, unless the Min-Plus Convolution Hypothesis is false, the $\grsp$ problem cannot be solved \emph{exactly} in $O(mn^{1-\delta})$ time for any $\delta > 0$ even in its most natural and simple special case.

\paragraph{Approximate $\grsp$.} The hardness results motivate looking for approximate algorithms for the $\grsp$ problem, especially if we seek to go beyond $\Theta(mn)$ runing time.

Indeed, the $\grsp$ problem has previously gained a lot of attention in the approximate setting~\cite{Bernstein12, goel2001efficient, Hassin92, lorenz2001simple, Warburton87}.
To discuss the previous work in more detail, let us first introduce some more notation.
\begin{definition}
Let $P^*$ be the shortest $(s,t)$-path with delay bounded by $D$ in $G$.
We say that an $(s,t)$-path $P\subseteq G$ constitutes an \emph{$(\alpha,\beta)$-approximate} solution to the $\grsp$ problem, if its length satisfies $\ell(P)\leq \alpha\cdot \ell(P^*)$ and its delay satisfies $d(P)\leq \beta\cdot D$.

We call an $\grsp$ algorithm $(\alpha,\beta)$-approximate if it produces $(\alpha,\beta)$-approximate solutions.
\end{definition}

\cite{goel2001efficient} gave an $\Ot(mn/\eps)$-time $(1,1+\eps)$-approximate algorithm for the single-source $\grsp$ problem. Their result leverages 
the observation that simple dynamic programming can, in fact, solve any instance of $\grsp$ with integral delays in $O(mD)$ time~\cite{joksch1966shortest}, and combines it with delay rounding/scaling in a standard way.
\cite{lorenz2001simple} showed a $(1+\eps,1)$-approximate algorithm for the single-source single-target case running in $\Ot(mn/\eps)$ time; this essentially follows by swapping the roles of lengths and delays and proceeding similarly\footnote{The algorithm of \cite{lorenz2001simple} requires some extra care to make it run in strongly polynomial time. This is also the reason for the single-target requirement.}
to~\cite{goel2001efficient}.  

Specifically for \emph{undirected} graphs,~\cite{Bernstein12} obtained a $(1+\eps,1+\eps)$-approximate single-source $\grsp$ algorithm running in $\Ot(m)\cdot (2/\eps)^{O\left(\sqrt{\log{n}}\log\log{n}\right)}$ time.
Note that for $\eps=\Omega(1)$, this bound is $mn^{o(1)}$, i.e., almost linear.
Apart from the rounding/scaling DP,~\cite{Bernstein12} exploits an emulator construction that crucially relies on the input graph being undirected.

Interestingly, the two $\Ot(mn/\eps)$ time algorithms~\cite{goel2001efficient, lorenz1998qos} remain the state-of-the-art approximate solutions for \emph{directed graphs} to date. 

\subsection{Our Results}
In this paper, we design the first approximate algorithms for restricted shortest paths in directed graphs polynomially improving upon the $\tilde{O}(mn)$ bounds.
Similarly to~\cite{Bernstein12, goel2001efficient}, our main focus is on the single-source variant defined formally below.
\begin{problem}[$\rsp$]\label{prob:rsp}
    Given a directed graph $G = (V, E, \ell, d)$, a source $s \in V$, and a delay threshold $D$, output for all $t \in V$ the length of the shortest $(s,t)$-path whose delay is at most $D$.
\end{problem}
In~\Cref{prob:rsp}, we concentrate only on returning path lengths for simplicity. All our algorithms for $\rsp$ can be extended to produce some actual $(s,t)$-path achieving the guaranteed lengths and delays with near-optimal overhead. The details can be found in~\Cref{sec:returning-paths}.

First of all, we give a near-optimal approximate\footnote{For simplicity of the stated bounds, in this paper we assume $\eps^{-1}=\poly(n)$. If $\eps^{-1}$ is larger than polynomial in~$n$, then neither the state-of-the-art nor our algorithms run in polynomial time.} algorithm for $\rsp$ in dense graphs:
\begin{theorem}[$\rsp$ on Dense Graphs; simplified version of \Cref{thm:dense} in~\Cref{sec:dense}]\label{thm:dense-simple}
    There is a Monte Carlo randomized $(1 + \eps, 1 + \eps)$-approximate algorithm for $\rsp$ that runs in $\Ot(n^2\log(W)/\eps^3)$ time, where~$W$ is the aspect ratio\footnote{Formally, the aspect ratio $W$ is the quantity $(\max_{e\in E_{\ell+}}\ell(e))/(\min_{e\in E_{\ell+}}\ell(e))$, where $E_{\ell+}$ is the subset of edges with positive lengths. In particular, the zero lengths do not influence the aspect ratio. It is worth noting that by swapping the roles of lengths and delays, the running time could be logarithmic in the aspect ratio of delays instead.} of lengths.
\end{theorem}

Note that if the aspect ratio is quasipolynomial in $n$ (say, if the weights come from the interval $[1,n^{\log{n}}]$), this improves upon the best known approximate algorithms \cite{goel2001efficient, lorenz2001simple} by a polynomial factor for all but sparse graphs with $m=\Ot(n)$. This, however, comes at the cost of being approximate in both the length and the delay,
a phenomenon that also manifested in~\cite{Bernstein12}.

What is also significant is that the algorithm is faster by a subpolynomial factor than the state-of-the-art~\cite{Bernstein12} (which gives the same approximation guarantee) even for dense \emph{undirected} graphs.

We also show an algorithm better suited to sparse and moderately dense graphs that improves upon~\Cref{thm:dense-simple} whenever $m=O(n^{3/2})$.
\begin{theorem}[$\rsp$ on Sparse Graphs; simplified version of \Cref{thm:sparse} in~\Cref{sec:sparse}]\label{thm:sparse-simple}
Let $m = n^{1 + \alpha}$, for $\alpha \in [0,1/2]$.
    There is a Monte Carlo randomized $(1+\eps,1+\eps)$-approximate algorithm for $\rsp$ that runs in $\Ot\paren{mn^{(3 - \alpha)/5} \log(W) / \eps^4}$ time, where $W$ is the aspect ratio of lengths.
\end{theorem}
A user-friendly way to read the above running time is $\Ot\paren{(mn)^{4/5}}$ or $\Ot\paren{mn^{3/5}}$-or-better when $m = O(n^{3/2})$, but we keep the extra parameter $\alpha$ around to make our proofs cleaner.

In particular, for $\eps^{-1}=n^{o(1)}$ and quasipolynomial aspect ratio, our algorithm for sparse graphs improves upon the state-of-the-art approximation schemes for directed graphs~\cite{goel2001efficient, lorenz2001simple}  for the entire range of possible graph densities. 

While the error in both the lengths and delays is inherent to our approach to $\rsp$, the dependence on the aspect ratio in \Cref{thm:dense-simple,thm:sparse-simple} can be removed without overhead if we only care about solving the problem for a single source-target pair $(s,t)$ (see~\Cref{lem:st-strong-poly}). It is worth noting that \cite{lorenz2001simple} also needed the single source-target pair requirement to remove the dependence on the aspect ratio in their $(1+\eps,1)$-approximate $\Ot(mn)$-time algorithm.

On the way to proving~\Cref{thm:all-pairs-simple}, we also show a near-optimal algorithm for the all-pairs variant of $\rsp$, which we believe is of independent interest.
\begin{theorem}[All-Pairs $\rsp$; \Cref{thm:all-pairs}~in~\Cref{sec:all-pairs} simplified]\label{thm:all-pairs-simple}
There is a Monte Carlo randomized $(1,1+\eps)$-approximate algorithm that solves $\rsp$ for all sources $s\in V$ at once in
$\Ot(mn/\eps+n^2/\eps^2)$ time.
\end{theorem}
Note that the $\Omega(mn)$ term is required in~\Cref{thm:all-pairs-simple} simply because for $D=\infty$ the problem being solved is equivalent to the standard exact APSP problem, for which it is conjectured there is no truly subcubic algorithm.
To the best of our knowledge, no non-trivial solutions for (approximate) All-Pairs $\rsp$ have been described so far.

\subsection{Techniques and Organization}

\paragraph{Techniques.}
We now give a brief summary of our main techniques; see~\Cref{sec:overview} for a more detailed high-level overview.

We start with a dynamic programming approach to solving $(1,1+\eps)$-approximate \linebreak $\rsp$ in~\Cref{sec:dp}.
The dynamic program, called the $\pidp{\pi}$, can be thought of as a generalization of the rounding/scaling dynamic programs leveraged in~\cite{Bernstein12, goel2001efficient}.
Crucially, our variant takes as input, in addition to the graph $G$ and source $s$, a \emph{frequency function} $\pi:E\to \mathbb{N}$.
The function $\pi$ defines, roughly speaking, how frequently the individual edges should be inspected by the DP to keep the error bounded.
If the frequency function is chosen suitably, the DP can run substantially faster than $\Ot(mn)$ time.
For example, a good $\pi$ function can be chosen very easily for DAGs, so that one obtains an $\Ot(n^2)$ time $(1,1+\eps)$-approximate algorithm for $\rsp$ in acyclic graphs. To the best of our knowledge, no specialized algorithms for DAGs beyond~\cite{goel2001efficient} have been described in the prior literature.

Our main results for $\rsp$ in general graphs are obtained, roughly speaking, by computing a good frequency function to be fed into the $\pidp{\pi}$. A good $\pi$ function for the input graph~$G$ may not necessarily exist or can be difficult to compute. We overcome this difficulty by also computing a certain auxiliary subgraph $H$ such that adding $H$ to $G$ does not improve the optimal solutions in~$G$; we then show how to compute a good frequency function $\pi$ for $G \cup H$. We show two entirely different constructions of $\pi, H$ for dense and sparse graphs, leading to our two main results (\Cref{thm:dense-simple,thm:sparse-simple}). In both cases, we start with a relatively simple construction of $\pi, H$ for DAGs, and then show how to extend these to general graphs via directed Low-Diameter Decompositions~\cite{BernsteinGW20,BernsteinNW22}, among other tools.

\paragraph{Related Work.} The high-level idea of using a frequency function to guide the underlying algorithm was also used in the very recent combinatorial max-flow algorithm of \cite{BernsteinBST24}; the details of computing this frequency function are very different, however, simply because max-flow and $\rsp$ are such different problems.

Our algorithm bears more similarity to the decremental shortest path algorithm of \cite{BernsteinGW20}, whose \emph{approximate topological ordering} is similar in spirit to a frequency function. Our two main results (\Cref{thm:dense-simple,thm:sparse-simple}) take as their starting point simple frequency functions for dense and sparse DAGs respectively, and these constructions are quite similar to the DAG algorithms of \cite{BernsteinGW20}, although as we discuss later, our interface of a frequency function gives us additional versatility that leads to a slight polynomial improvement for sparse DAGs.

Our algorithms for general directed graphs, however, are quite different from those of \cite{BernsteinGW20}, largely because the underlying problems are different ($\rsp$ vs. dynamic shortest paths). We both use a hierarchy of directed LDDs, but the constructions are very different. In particular, our construction is top-down whereas that of \cite{BernsteinGW20} is bottom-up; this leads to different interactions between the levels of the hierarchy, and hence a completely different analysis.

\paragraph{Organization.} %

\Cref{sec:preliminaries} covers the preliminaries.
\Cref{sec:dp} introduces the $\pidp{\pi}$ algorithm, which is the foundation for all our results.
In~\Cref{sec:overview} we provide a high-level overview of our main results, with the goal of providing some intuition and easing the reader into the more technical sections to follow.
In~\Cref{sec:framework} we formally reduce the problem of $\rsp$ to the problem of finding auxiliary edges $H$ and a good $\pi$ function on $G\cup H$.
\Cref{sec:dense,sec:sparse} describe how the $\pi,H$ pairs can be found efficiently in dense and sparse digraphs, respectively, thus completing the corresponding proofs of \Cref{thm:dense-simple,thm:sparse-simple}.
\Cref{sec:all-pairs} is devoted to our All-Pairs $\rsp$ algorithm of~\Cref{thm:all-pairs-simple}, obtained by combining dynamic programming with the shortcutting approach used in~\cite{Bernstein16} for decremental APSP.
Finally, in~\Cref{sec:open-problems} we discuss some directions for further work on $\rsp$.

\section{Preliminaries}
\label{sec:preliminaries}

\paragraph{Numbers and Sets. }
For a positive integer $n \in \IN$, we denote $\set{1, 2, \ldots, n}$ with $[n]$.
When the context is clear, for integers $a,b \in \IN$ we use $[a, b]$ to denote $\set{a, a+1, \ldots, b}$ (rather than all real numbers between $a$ and $b$).
For positive integers $k < n$, we write $k | n$ to indicate that $k$ divides $n$.

\paragraph{Graphs. }
Throughout the paper, we consider directed graphs $G = (V, E)$.
When it is clear, we use $n = \card{V}$ and $m = \card{E}$.
To suppress uninteresting details, we assume that $G$ is given as a simple graph\footnote{Our results, nevertheless, hold for $G$ given with parallel edges, which is an especially relevant concern for $\rsp$ seeing that choosing the right edge to take between a pair of vertices explores a tradeoff between the length and delay available to the rest of the delay restricted path. We explain how to address this in \Cref{sec:parallel-edges}.}, but our techniques then modify $G$ so that it may contain parallel edges.
We denote an edge connecting vertices $u$ to $v$ with $uv$; typically, it will be unambiguous which edge we are referring to when there are multiple edges directed from $u$ to $v$.
A path $P$ is a sequence of vertices $v_0, v_1, \ldots, v_h$ such that $v_i v_{i+1}$ are edges for all $i \in [h - 1]$; we say that $P$ is a $h$-hop path.
We also call $P$ a $(v_0,v_h)$-path.
We usually treat $P$ as the set of its edges, and refer to a set of contiguous edges as a subpath of $P$.
Sometimes, we may intersect a graph with an object that is unambiguously a set of edges; in this case, we are also treating the graph as the set of its edges.
If we write $s \in V$ or $t \in V$ and do not qualify them further, they stand for ``source'' and ``target'' respectively.

The graphs in this paper will come equipped with lengths $\ell: E \rightarrow \IR_{\ge 0}$ and delays
${d: E \rightarrow \IR_{\ge 0}}$, and sometimes weights $w: E \rightarrow \IR_{\ge 0}$.
The lengths, delays, and weights of a set of edges (in particular, a path) $P$ are written with $\ell(P) = \sum_{e \in P} \ell(e)$, $d(P) = \sum_{e \in P} d(e)$, and 
$w(P) = \sum_{e \in P} w(e)$, respectively.

For $C \subseteq V$, we use $G[C]$ to denote the subgraph of $G$ induced on $C$.
For $H \subseteq E$ we use $G \setminus H$ to denote the subgraph of $G$ attained by removing $H$ and $G \cup H$ to denote the supergraph of $G$ attained by adding $H$.
We often write $H\cap G[C]$ as shorthand for $H\cap E(G[C])$.

For $u,v\in V$ and $D\in\mathbb{R}_{\geq 0}$, let us denote by $\rdist(s,t,D)$ the smallest possible length of an $(s,t)$-path $P$ in $G$ satisfying $d(P)\leq D$.
When considering delay-constrained paths $P$ in a different graph $H \neq G$, we may sometimes use the notation $\rdist_H(s,t,D)$.

The strongly connected components (SCCs) of a directed graph $G$ are maximal vertex sets $C \subseteq V$ such that, for all $u, v \in C$ there is a path in $G$ from $u$ to $v$.
A classic in the algorithm design canon shows that we can compute the SCCs in linear time and, combining with another classic, they can be presented in topological order in the same time.
\begin{proposition}[Theorems 13, 14 in \cite{tarjan1972depth} and \cite{tarjan1976edge}]
\label{prop:scc}
    The strongly connected components of a directed graph $G$ can be computed in topological order of the DAG where each SCC is contracted in $O(n + m)$ time.
\end{proposition}

\paragraph{Pareto Frontier. }

Refer to \Cref{fig:pareto} for this paragraph, and a depiction of $(1+\eps, 1+\eps)$-approximate solutions.
We say that an $(s,t)$-path $P$ Pareto dominates an $(s,t)$-path $P'$ if \linebreak
${\ell(P) \le \ell(P')}$ and $d(P) \le d(P')$, and that $P$ weakly Pareto dominates $P'$ if at least one of the inequalities is an equality; visually, the point on a length-delay plot corresponding to $P'$ is northeast to that of $P$.
For an $(s,t)$-path $P$, let $R_P = \set{(x,y): x \ge d(P) \textrm{ and } y \ge \ell(P)}$.
The Pareto Frontier for $s,t$ is the boundary of $\bigcup_{(s,t)\textrm{-paths } P} R_P$.

\begin{figure}[h]
    \centering
    \includegraphics[scale=0.6]{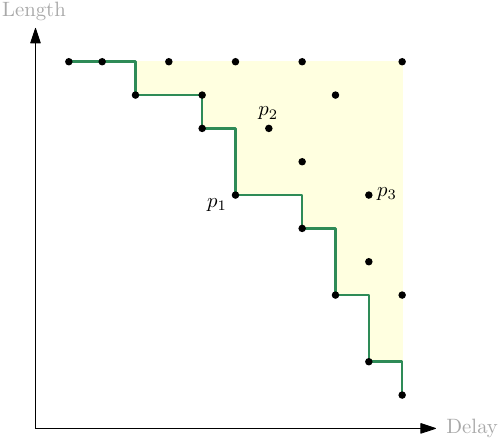}
    \includegraphics[scale=0.6]{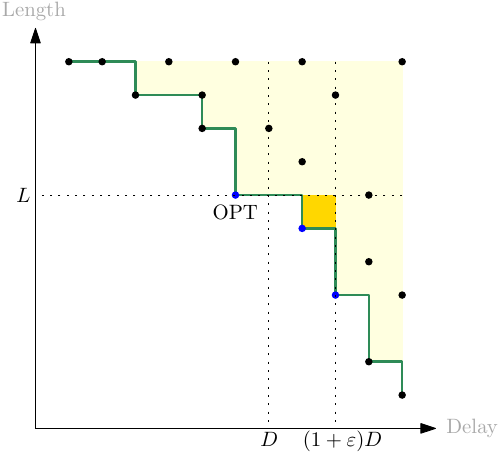}
    \includegraphics[scale=0.6]{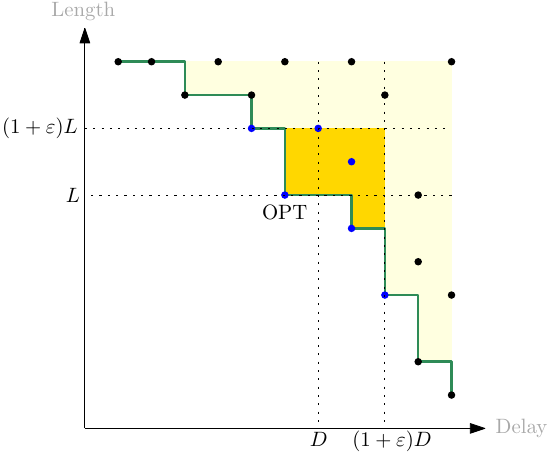}
    \caption{
        Length-delay plots of $(s,t)$-paths for a fixed $s$ and $t$.
        Green line indicates the Pareto Frontier.
        Left:
        $p_1$ Pareto dominates $p_2$ since $p_2$ lies to the northeast of $p_1$, and (weakly) Pareto dominates $p_3$ since $p_3$ lies to the east of $p_1$.
        Neither $p_2$ nor $p_3$ Pareto dominate each other as neither are in the northeast region of the other.
        Center: The gold region and blue points indicate valid answers a $(1, 1+\eps)$-approximation can give.
        Right: The gold region and blue points indicate valid answers a $(1+\eps, 1+\eps)$-approximation can give.
    }
    \label{fig:pareto}
\end{figure}

\paragraph{Directed Low-Diameter Decomposition. }
Our results utilize directed Low-Diameter Decompositions (LDDs), an efficient computation of which has recently been used in \cite{BernsteinNW22} for the Negative-weight Single Source Shortest Paths problem.
An LDD is, at a high level, a set of edges $B$ such that, when removed from $G$, the SCCs of $G$ have bounded diameter.
For clarity of exposition, we use the definition of LDD in \cite{bringmann2023negative} which gives a stronger notion of LDDs compared to what is used in \cite{BernsteinNW22}.
\begin{definition}[(Strong) Directed Low-Diameter Decomposition]
\label{def:ldd}
    Given a directed graph $G$ with non-negative edge weights $w$ and a parameter $D > 0$,
    a Low-Diameter Decomposition with overhead~$\rho$ 
    is a random edge set $B \subseteq E(G)$ with the following properties:
    \begin{itemize}
        \item Sparse Hitting: For any edge $e \in E(G), \prob{e \in B} = O\paren{\frac{w(e)}{D}\rho + \frac{1}{\poly(n)}}$ where the degree of the polynomial is an arbitrarily large constant.
        \item Bounded Diameter: For any strongly connected component $C$ in $G \setminus B$, for any $u, v \in C$, we have  $\dist_{G \setminus B}(u,v) \le D$.
    \end{itemize}
\end{definition}
Sparse Hitting is a desirable property for the following reason: Suppose there is a path $P$ with small weight with respect to the diameter bound.
By linearity of expectation, only a small number of edges in $P$ are expected to be a part of $B$, the LDD output.
Then, as we walk along the projection of $P$ onto $G$ where the SCCs of $G \setminus B$ have been contracted, most of the steps are increasing in the topological order of the SCCs. This imposes an additional structure that can be exploited algorithmically.

We make use of the near-linear time construction of strong LDDs from \cite{bringmann2023negative}.%
\begin{proposition}[Theorem 5 in \cite{bringmann2023negative}]
\label{prop:ldd}
    There is a Low-Diameter Decomposition with overhead $O(\log^3 n)$, computable in time $O\paren{(m + n \log\log n)\log^2 n}$ with high probability (and in expectation).
\end{proposition}

\section{A Frequency-Driven Dynamic Program}
\label{sec:dp}
The starting point of our results -- and perhaps the most crucial building block -- is a dynamic program  whose running time is driven by a function $\pi : E \rightarrow \IN$ which defines the frequency with which each edge is inspected.
For simplicity, let us assume that frequencies are at most polynomial in $n$. 
Differently defined $\pi$'s yield algorithms with different guarantees and, as we will see later, the crux of our results boils down to efficiently finding a good~$\pi$.
As a first pass, it is instructive to think of $\pi$ as the function that is constantly $1$.

We will use the dynamic program to obtain $(1,1+\eps)$-approximate solutions to the $\rsp$ problem for arbitrary delays $D$ in an input interval $[D_{\min},D_{\max}]$ such that $D_{\max}/D_{\min}=\poly(n)$.
To state the guarantees that the DP provides formally, we need to introduce the following important notation.
For any path (or more generally, subset of $E$) $P$, let $\pi(P) = \sum_{e \in P} \pi(e)$.
Let $\rdist^{\pi\leq h}(s,t,x)$ denote the smallest length of an $(s,t)$-path~$P$ in $G$ such that $\pi(P)\leq h$ and $d(P)\leq x$.
Moreover, in the special case when $\pi(e)=1$ for all $e\in E$, let us put $\rdist^h(s,t,x):=\rdist^{\pi\leq h}(s,t,x)$.

In~\Cref{sec:dp-proof}, we prove the following.

\begin{theorem}\label{thm:pi-dp}
    Let $\pisum:=\sum_{e\in E}\pi(e)=O(\poly(n))$ and let $\Pi = \sum\limits_{e \in E}\frac{1}{\pi(e)}$.
    Let $h\in [1,\pisum]$ be an integer. Let $D_{\min}\leq D_{\max}$ be such that $D_{\max}/D_{\min}=\poly(n)$. Let $s\in V$ be a source.
    
    After $O\paren{h \cdot \Pi \log(n)/\eps}$
    time preprocessing, for any $t\in V$ and $D\in[D_{\min},D_{\max}]$, in constant time we can compute the value $\ell(P_t)$, where $P_t$ is some $(s,t)$-path in $G$ satisfying:
    \begin{itemize}
        \item $\ell(P_t)\leq \rdist^{\pi\leq h}(s,t,D)$,
        \item $d(P_t) \le (1 + \eps) D$.
    \end{itemize}
\end{theorem}

In the following, for a frequency function $\pi: E \rightarrow \IN$, we denote the algorithm of \Cref{thm:pi-dp} with $\pidp{\pi}$. We also sometimes call the $h$ parameter in~\Cref{thm:pi-dp} the \emph{depth} of $\pidp{\pi}$.

\subsection{Examples of Useful Frequencies}
\label{subsec:examples}
Before we prove~\Cref{thm:pi-dp}, we state two easy but important frequency functions; the first will in fact be used as a subroutine in one of our main results (see \Cref{thm:sparse-simple}). The second provides motivation to our main results (see \Cref{thm:dense-simple,thm:sparse-simple}) and, at a more basic level, gives a quick demonstration of why $\pidp{\pi}$ is a promising idea.

\paragraph{Example 1: Number of Hops. }
Let $\one: E \rightarrow \set{1}$ be the constantly $1$ function.
Then the $\pidp{\one}$ can be used to find, in $\Ot(mh/\eps)$ time, for all $t \in V$, the length of a path $P_t$ satisfying $\ell(P_t)\leq\rdist^h(s,t,D)$ and $d(P_t) \le (1 + \eps) D$.
In particular, for $h=n$, the $\pidp{\one}$ solves the \linebreak $(1,1+\eps)$-approximate
$\rsp$ in $\Ot(mn/\eps)$ time.

The $\pidp{\one}$ subroutine is a key ingredient to implementing an $\Ot(mn/\eps + n^2/\eps^2)$ time
\linebreak $(1, 1 + \eps)$-approximate algorithm for the All-Pairs $\rsp$ (see \Cref{thm:all-pairs-simple} and \Cref{sec:all-pairs}), which in turn is used as a subroutine for \Cref{thm:sparse-simple}.

\paragraph{Example 2: Topological Order Difference. }
Suppose $G$ is a DAG, and let $\tau: V \rightarrow [n]$ define a topological order on its vertices.
We extend $\tau$ to edges as follows: $\tau(uv) = \tau(v) - \tau(u)$.
Then $\pidp{\tau}$ run with $h=n$ is a $(1,1 + \eps)$-approximate algorithm for $\rsp$ that runs in $\Ot(n^2/\eps)$ time.
To see that the running time is $\Ot(n^2/\eps)$, observe that $\tau(P) \le n$ for all paths $P$ since $G$ is a DAG and $\tau$ is a topological order.
Next, observe that
\begin{align*}
    \Pi
    = \sum_{e \in E} \frac{1}{\tau(e)} 
    = \sum_{u \in V} \sum_{\substack{v \in V:  uv \in E}} \frac{1}{\tau(v) - \tau(u)}
    \le \sum_{u \in V} \sum_{i \in [n]} \frac{1}{i}
    = O(n \log n).
\end{align*}
Putting the two observations together shows that $\pidp{\tau}$ runs in $\Ot(n^2/\eps)$ time on DAGs.
Our main results extend this idea to directed graphs that are not necessarily acyclic.

\subsection{Details of the Dynamic Program}
\label{sec:dp-proof}
The reader who is willing to take the statement of \Cref{thm:pi-dp} as is, without any further justification, may skip this subsection and go directly to \Cref{sec:overview} which gives an overview of our main techniques.

We obtain~\Cref{thm:pi-dp} via dynamic programming.
Roughly speaking, the dynamic program will compute, for a fixed source $s$ and all vertices $t \in V$, the length of the shortest $(s,t)$-path $P$ whose \emph{$\eps$-rounded delay} (which is a good proxy for the delay) is bounded above by a threshold.
First, let us say more about what we mean by $\eps$-rounded delay.

\paragraph{Edge Delay Rounding. }
To avoid dependence on the aspect ratio of delays\footnote{Though a dependence on the aspect ratio of lengths will sadly creep into our main results later.}, the rounded delay bumps small delays up to
\begin{align*}
    \delta \coloneqq \frac{\eps}{\pisum}D_{\min}.
\end{align*}
That is, the rounded delay of an edge $e$ with delay $d(e)$ is
\begin{align*}
    \bump(e) = \max(d(e), \delta).
\end{align*}
Observe that the ratio between $D\in [D_{\min},D_{\max}]$ and $\min_e(\bump(e))$ is then at most $(D_{\max}/D_{\min})\eps/\pisum$, which for our purposes is small.

\paragraph{Path Delay Rounding. }
For any real $\gamma > 1$ and $x > 0$, we define $\expround{\gamma}{x}$ to be the smallest integer power of $\gamma$ that is at least as large as $x$, and $\expround{\gamma}{0} = -\infty$.
Then, the $\eps$-rounded delay of a path $P$ is defined inductively as so:
\begin{align*}
    d'(P) = 
        \begin{cases}
            0 & \textrm{if } P = \emptyset \\
            \expround{(1 + \eps)^{\pi(e)}}{d'(P') + \bump(e)} & \textrm{if } P = P' \cdot e
        \end{cases}
\end{align*}
For reasons that will become clear very shortly, take note that if $d'(P' \cdot e) = (1 + \eps)^{i}$, then $i$ is an integer multiple of $\pi(e)$.

\paragraph{The Dynamic Program. }
Let $\pi(P) = \sum_{e \in P} \pi(e)$ for any path $P$.
For all $t \in V$ and ${k \in \IZ \cup \set{-\infty}}$, define
\begin{align*}
    \dptable{t}{k} \coloneqq & \textrm{ the length of the shortest } (s,t)\textrm{-path } P \\& \textrm{ such that } d'(P) \le (1 + \eps)^k,
\end{align*}
with boundary values $\dptable{s}{k} = 0$ for all $k$ (i.e. empty paths are paths) and $\dptable{t}{i} = \infty$ for all $t \neq s$ and $(1 + \eps)^i < \delta$ (since $d'(P) \ge \delta$ for non-empty $P$).

The non-boundary values of $\dptable{\cdot}{\cdot}$ can be computed by the following recurrence:

\begin{tbox}
    \algname{Recurrence for $\dptable{\cdot}{\cdot}$}
    \[\dptable{t}{k} = \min\paren{
        \dptable{t}{k-1},
        \min_{\substack{ut \in E :\\ \pi(ut)|k}}\paren{
            \dptable
            {u}
            {\floor{\log_{(1 + \eps)}\paren{(1 + \eps)^k - \bump(ut)}}}
            + \ell(ut)
        }
    }\]
\end{tbox}

\begin{proof}[Proof of Recurrence]
    We give a proof by induction on the number of edges in a path.
    \begin{itemize}
        \item Base Case:
            If $P = st = e$ is the shortest path such that $d'(P) \le (1 + \eps)^k$, then $d'(P) = \expround{(1 + \eps)^{\pi(e)}}{\bump(e)} = (1 + \eps)^i$ where $i \le k$ is an integer multiple of $\pi(e)$.
            The length $\ell(P)$ will thus be stored in $\dptable{t}{i}$, which is computed by accessing a boundary value, and this will be propagated to $\dptable{t}{k}$ by the first argument of the outer minimum function.
        \item Induction Step:
            If $P = P' \cdot ut$ for a non-empty $P'$ is the shortest path such that $d'(P) \le (1 + \eps)^k$, then $d'(P) = \expround{(1 + \eps)^{\pi(ut)}}{d'(P') + \bump(ut)} = (1 + \eps)^i$ where $i \le k$ is an integer multiple of $\pi(e)$.
            We have \[d'(P') \le (1 + \eps)^\floor{\log_{(1 + \eps)}\paren{(1 + \eps)^i - \bump(ut)}}\] and hence, by the induction hypothesis, $\ell(P')$ will be stored in $\dptable{u}{\floor{\log_{(1 + \eps)}\paren{(1 + \eps)^i - \bump(ut)}}}$.
            The length $\ell(P)$ will thus be stored in $\dptable{t}{i}$, and this will be propagated to $\dptable{t}{k}$ by the first argument of the outer minimum function.
    \end{itemize}
\end{proof}

\paragraph{Finding Lengths of Shortest Paths With Bounded Delay. }
To conclude this subsection, we show how to efficiently compute and then transfer the values of $\dptable{\cdot}{\cdot}$, which are based on $\eps$-rounded delays $d'$, to lengths of shortest paths whose true delays $d$ are bounded by $(1 + \eps) D$.

First, let us see why $\eps$-rounded delays are a good proxy for delays.
\begin{observation}
\label{obs:rounded-delay-is-proxy}
    For any $(s,t)$-path $P\subseteq G$, the $\eps$-rounded delay satisfies the following:
    \begin{align*}
        d'(P) \le (1 + \eps)^{\pi(P)}\paren{d(P) + \pi(P) \cdot \delta}
    \end{align*}
\end{observation}
\begin{proof}
    The proof is by induction on the number of edges in $P$.
    The base case is immediate.
    Let $P = P' \cdot e$.
    The induction step is completed by
    \begin{align*}
        d'(P' \cdot e)
        & =
        \expround{(1 + \eps)^{\pi(e)}}{d'(P') + \bump(e)}
        \\
        & \le
        \expround{(1 + \eps)^{\pi(e)}}{(1 + \eps)^{\pi(P')}\paren{d(P') + \pi(P') \cdot \delta} + \bump(e)}
        \tag{Induction Hypothesis}
        \\
        & \le
        \expround{(1 + \eps)^{\pi(e)}}{(1 + \eps)^{\pi(P')}\paren{d(P') + \pi(P') \cdot \delta} + d(e) + \delta}
        \tag{$\bump(e) \le d(e) + \delta$}
        \\
        & \le
        \expround{(1 + \eps)^{\pi(e)}}{(1 + \eps)^{\pi(P')}\paren{d(P) + \pi(P) \cdot \delta}}
        \tag{$1 \le \pi(e)$ and $1 \le (1 + \eps)^{\pi(P')}$}
        \\
        & \le
        (1 + \eps)^{\pi(P)} \paren{d(P) + \pi(P) \cdot \delta}.
    \end{align*}
\end{proof}

We can now use \Cref{obs:rounded-delay-is-proxy} to make a connection between the $\eps$-rounded delays and delays in $\dptable{\cdot}{\cdot}$, which will then be useful for extracting the output for a $(1,1+\eps)$-approximate solution to $\rsp$.

\begin{lemma}
\label{lem:delay-in-dp-table}
    Let $P$ be a simple $(s,t)$-path in $G$. If $D\in [D_{\min},D_{\max}]$ and $d(P) \le D$, then the following~holds:
    \begin{align*}
        \dptable{t}{\ceil{\log_{1 + \eps}\paren{(1 + \eps)D}} + \pi(P)} \le \ell(P).
    \end{align*}
\end{lemma}
\begin{proof}
    Let $P$ be a path such that $\pi(P) \le \pisum$.
    Then the rounded delay of $P$ is bounded by:
    \begin{align*}
        d'(P)
        & \le
        (1 + \eps)^{\pi(P)}\paren{d(P) + \pi(P) \cdot \delta}
        \tag{\Cref{obs:rounded-delay-is-proxy}}
        \\
        & \le
        (1 + \eps)^{\pi(P)}\paren{D + \pisum \cdot \delta}
        \tag{$\pi(P) \le \pisum$ and $d(P) \le D$}
        \\
        & =
        (1 + \eps)^{\pi(P)}\paren{D + \eps D_{\min}}
        \tag{$\delta = \frac{\eps}{\pisum}D_{\min}$}
        \\
        & =
        (1 + \eps)^{\pi(P)}\paren{D + \eps D}
        \tag{$D\geq D_{\min}$}
        \\
        & \le
        (1 + \eps)^{\ceil{\log_{1 + \eps}\paren{(1 + \eps)D}} + \pi(P)}.
    \end{align*}
    By definition, $\dptable{t}{\ceil{\log_{1 + \eps}\paren{(1 + \eps)D}} + \pi(P)}$ is the length of the shortest path with $\eps$-rounded delay bounded above by $(1 + \eps)^{\ceil{\log_{1 + \eps}\paren{(1 + \eps)D}} + \pi(P)}$ for which, by the above inequality, $\ell(P)$ is a candidate.
\end{proof}

We are now ready to prove~\Cref{thm:pi-dp}.

\begin{proof}[Proof of~\Cref{thm:pi-dp}]
    Let $\eps' = \eps / (2h+4)$.
    We compute the non-boundary values  $\dptable{t}{k}$ wrt. error parameter $\eps'$ (instead of just $\eps$, which also means redefined $\delta:=(\eps'/\pisum)D_{\min}$) for all $t \in V$, where \[k \in \bracket{\floor{\log_{1 + \eps'} \delta}, \ceil{\log_{1 + \eps'}\paren{(1 + \eps')D_{\max}}} + h}.\]
    (Values $\dptable{t}{k}$ for $k$ below this range are boundary values.)
    To produce the desired value for a query $(t,D)$, we simply return 
    $\dptable{t}{\ceil{\log_{1 + \eps'}\paren{(1 + \eps')D}} + h}$.
    
    \underline{Running Time:}
    First, observe that there are
    \[ O(\log((1+\eps)D_{\max}/\delta)/\eps'+h)=O(\log((D_{\max}/D_{\min})\pisum/\eps')/\eps'+h)=O\paren{h \cdot \log(n)/\eps} \]
    values that $k$ takes in the above-stated range.
    An edge $e$ is inspected for an $O(1 / \pi(e))$ fraction of the values $k$ can take, settling the running time by summing over all edges.

    \underline{Correctness:}
    For an arbitrary $t$, let $P_t^*$ be any $(s,t)$-path satisfying $\pi(P_t^*) \le h$ and $d(P_t^*)\leq D$.
    By~\Cref{lem:delay-in-dp-table},
    \[ \dptable{t}{\ceil{\log_{1 + \eps'}\paren{(1 + \eps')D}} + h} \le
    \dptable{t}{\ceil{\log_{1 + \eps'}\paren{(1 + \eps')D}} + \pi(P_t^*)}\ \leq
    \ell(P_t^*).\]
    Let $P_t$ be an $(s,t)$-path such that $\ell(P_t)=\dptable{t}{\ceil{\log_{1 + \eps'}\paren{(1 + \eps')D}} + h}$. By the arbitrary choice of $P_t^*$ (wrt. length), we have established $\ell(P_t)\leq \rdist^{\pi\leq h}(s,t,D)$.
    Note that the delay of $P_t$ is:
    \begin{align*}
        d(P_t)
        & \le
        d'(P_t)
        \\
        & \le
        (1 + \eps')^{\ceil{\log_{1 + \eps'}\paren{(1 + \eps')D}} + h}
        \tag{By the definition of $\dptable{\cdot}{\cdot}$}
        \\
        & \le
        (1 + \eps')^h (1 + \eps')^2 D
        \\
        & =
        \paren{1 + \frac{\eps}{2(h+2)}}^{h + 2} D
        \\
        & \le 
        e^{\eps/2} D
        \\
        & \le
        (1 + \eps) D.
    \end{align*}
\end{proof}

\section{High-Level Overview}
\label{sec:overview}

In this section, we provide a sketch of some of the key ideas behind our results.
The exposition here is quite brief and informal and is only meant to broadly highlight the aforementioned ideas and provide some initial intuition.

\paragraph{Simplifying Assumptions.} 
For simplicity, we focus on a single source $s$ and single target $t$ and assume $\eps$ is a small fixed constant.
We also focus on a decision version of the problem where we are given a length $L$ and a delay $D$ and, loosely speaking, the goal is to determine if there exists an $(s,t)$-path of length $\lesssim L$ and delay $\lesssim D$ (we will elaborate on this shortly).
We assume for this overview that $L = D = n$; in \Cref{sec:framework}, we will show that this assumption can be achieved by rescaling lengths and delays.
We call this toy version of the problem Simplified $\rsp$.

\paragraph{Computing a Frequency Function.}
\Cref{thm:pi-dp} suggests that the basic dynamic programming algorithm can be made efficient if we first find a suitable frequency function $\pi$.
In particular, say that we are able to define $\pi$ with the following guarantee: if there exists an $(s,t)$-path $P$ with $\ell(P) \leq n$ and $d(P) \leq n$, then there exists an $(s,t)$-path $P'$ such that $\ell(P') \leq (1+\eps)n$ and $d(P') \leq (1+\eps)n$ and, moreover, $\pi(P') = \Ot(n)$.
If we have such a $\pi$ then, intuitively, we can solve Simplified $\rsp$ by running $\pidp{\pi}$ (\Cref{thm:pi-dp}) with depth $h = \Ot(n)$ and delay threshold $(1 + \eps)D$ (with some rescaling of $\eps$), and then finally checking if the output length is also no more than $(1 + \eps)n$.
By \Cref{thm:pi-dp}, the runtime is $\Ot(n \sum_e 1/\pi(e))$.
Our goal in finding a good function $\pi$ is thus to minimize $\sum_e 1/\pi(e)$ while ensuring that the above guarantee on $\pi$ still holds.

\paragraph{Adding Shortcut Edges.}
In fact we do not know how to find a good function $\pi$ directly for the graph $G$.
We will add some new edges $H$ to the graph $G$.
We construct $H$ so as to always satisfy: for any edge $uv \in H$, there exists a $(u,v)$-path $P_{uv}$ in $G$ with $\ell(P_{uv}) \leq \ell(uv)$ and $d(P_{uv}) \leq d(uv)$; this property ensures that for any path in $G \cup H$ there is a corresponding path in $G$ that is at least as good; we call such a $H$ \emph{Pareto Frontier Preserving}.

By allowing paths $P'$ in $G \cup H$, we will be able to control $\pi(P')$, which makes it easier to find a good function $\pi$.
Combining everything we have so far, we can show that Simplified $\rsp$ can be reduced to the problem of computing a good pair $\pi, H$.

\begin{proposition}[Simplified version of \Cref{lem:gap-solver} chained into \Cref{lem:reduction-to-gap}]
\label{prop:simplified-reduction}
    To solve Simplified $\rsp$, it is sufficient to compute a Pareto Frontier Preserving edge set $H$ and a frequency function $\pi$ on $G \cup H$ that together satisfy the \emph{Path-Sum Constraint}, which is the following property:
    \begin{quote}If there exists an $(s,t)$-path $P$ in $G$ with $\ell(P) \leq n$ and $d(P) \leq n$, then there exists an $(s,t)$-path $P'$ in $G \cup H$ such that $\ell(P') \leq (1+\eps)n$ and $d(P') \leq (1+\eps)n$ and $\pi(P') = \Ot(n)$.
    \end{quote}
    The running time of Simplified $\rsp$ is then $\Ot(n\cdot \Pi)$ plus the time taken to compute $\pi,H$, where $\Pi = \sum_{e \in G \cup H} 1/\pi(e)$.
\end{proposition}

\paragraph{Roadmap.}
The main technical contribution of the paper is showing how to efficiently compute $\pi, H$ that we can plug into \Cref{prop:simplified-reduction} (or its more technical analog in the later sections).
We use two different constructions: one for dense graphs and one for sparse graphs.
In both cases, we start by showing a relatively simple construction of $\pi, H$ for a DAG, and then show how to generalize this construction to arbitrary directed graphs using, among other tools, a directed Low-Diameter Decomposition.

\subsection{Computing $\pi,H$ in Dense Graphs}
In this subsection, we outline how to compute $\pi, H$ that satisfy the Path-Sum Constraint, and for which $\Pi = \sum_{e \in G \cup H} 1/\pi(e) = \Ot(n)$; plugging this into \Cref{prop:simplified-reduction} yields the $\Ot(n^2)$-time algorithm of \Cref{thm:dense-simple}.

\paragraph{Simple Case: Dense DAGs.}
Example 2 in \Cref{subsec:examples} already contains the desired $\pi$ when $G$ is a DAG: compute a topological order $\tau$ and set $\pi(uv) = \tau(v) - \tau(u)$ (in this case $H = \emptyset$).
It is easy to check that $\pi(P) \leq n$ for \emph{any} path $P$, so the Path-Sum Constraint is trivially satisfied.
As discussed in \Cref{subsec:examples}, we have $\Pi = \Ot(n)$, as desired.

\subsubsection{Dense Directed Graphs: First Attempt (One LDD)}
\label{sub:overview-dense-slow}
To compute $\pi, H$ for general graphs, we will again compute an ordering of the vertices $\tau: V \rightarrow [n]$ and set $\pi(uv) = \card{\tau(v) - \tau(v)}$.
But in general graphs, it is much less clear what the ordering should be, and it is more difficult to verify the Path-Sum Constraint: now that some edges can go backward in the ordering, $\pi(P)$ could be as large as $\Theta(n^2)$.

We will use an LDD to decompose the graph into low-diameter components and DAG edges.
Note, however, that the LDD only applies to a single edge weight function, while our edges have both a length and a delay.
To this end, we define: 

\begin{restatable}[Combined Length-Delay Graph]{definition}{comblendel}
\label{def:lendel}
    Let $\lendel{G} = (V, E, w = \ell + d)$, which we refer to also as the combined length-delay graph of $G$, be the graph which is topologically equal to~$G$ and equipped with the non-negative weight function $w = \ell + d$.
\end{restatable}

To get an intuition for how we use LDDs to obtain an ordering $\tau$, we start by showing how to compute a $\pi$ that satisfies the Path-Sum Constraint and has $\Pi = \Ot(n^{3/2})$.
We then show in the following subsubsection how to improve this to the desired $\Pi = \Ot(n)$.

\paragraph{Defining $H$.}
For dense graphs, our construction of the new edges $H$ is very simple.
Say that we compute an LDD of $\lendel{G}$ for some diameter bound $D$ (not to be confused with the delay threshold, which we have fixed to $n$ for Simplified $\rsp$).
For \emph{every} SCC $C$ of the resulting decomposition, we will pick an arbitrary representative $c \in C$, and then for every other $v \in C$ we will add edges $vc$ and $cv$ to $H$ (we refer to these as \emph{star}-edges).
Since $C$ has diameter $D$ in $\lendel{G}$, we can safely set $\ell(cv) = \ell(vc) = d(cv) = d(vc) = D$.
These star-edges will have relatively small $\pi$-value, so they help ensure the existence of a path $P'$ that satisfies the Path-Sum Constraint.
Note, however, that rerouting $P'$ along the star-edges can increase the cost and delay of $P'$, so although we add star-edges to every SCC $C$, the path $P'$ will only use a smaller number of them (more details below).

\paragraph{Constructing $\pi, H$.}
The construction of $\pi,H$ is as follows:
\begin{tbox}
\begin{enumerate}
    \item $\eback{} \gets \LDD(\lendel{G}, D = \sqrt{n})$. \textcolor{blue}{\texttt{($\eback{}$ is a set of edges)}}
    \item $\cC = \set{C_1, C_2, \ldots, C_k} \gets \SCC(G \setminus \eback{})$.
    \item For every $C_i$, add the star-edges described above to $H$.
    \item Let $\tau$ be the ordering such that:
        \begin{itemize}
            \item For any edge $uv \in G$ such that the SCC $u$ belongs to precedes the SCC $v$ belongs to in a topological ordering of SCCs in $\cC$, we have $\tau(u) < \tau(v)$.
            \item $\tau$ is contiguous in any SCC $C \in \cC$.
                That is: $\{\tau(v) : v \in C\} = [x,x+\card{C}-1]$ for some integer $x$.
            \item Set $\pi(uv) = \card{\tau(v) - \tau(u)} / \sqrt{n}$ (note the down-scaling\footnote{Even though~\Cref{thm:pi-dp} assumes integral $\pi$, non-integral frequencies are not a problem, see~\Cref{sec:framework}.}).
        \end{itemize}
\end{enumerate}
\end{tbox}

\paragraph{Analyzing $\pi$.}
It is not hard to check that $\Pi = \Ot(n^{3/2})$: the contribution of $H$ to $\Pi$ is negligible (because $\card{H} = O(n)$), and by the same reasoning as Example 2 of \Cref{subsec:examples}, we have $\sum_{uv \in E} 1 / \pi(uv) = \sqrt{n} \cdot \sum_{uv \in E} (1/\card{\tau(v)-\tau(u)}) = \Ot(n^{3/2})$.
All that remains is to verify that the Path-Sum Constraint is satisfied.

The LDD partitions the edges of $G$ into three types: \emph{intra-component} edges inside some $G[C_i]$, the inter-component edges of $\eback{}$, and finally the inter-component edges that obey the ordering~$\tau$ (denoted 
$\eforward{}$).
We will analyze the contributions of the three sets separately.

\underline{Intra-Component Edges}:
As assumed by the Path-Sum Constraint, let $P$ be an $(s,t)$-path with $\ell(P) \leq n$ and $d(P) \leq n$.
For any $C_i \in \cC$, $P$ could go in and out of $C_i$ many times, but since we can assume without loss of generality that $P$ is simple, we know that $\card{P \cap G[C_i]} \leq \card{C_i}$.
Because~$\tau$ is contiguous, we know that $\pi(uv) \leq \card{C_i}/\sqrt{n}$ for all $uv \in E(G[C_i])$, so all in all we have 
\[\pi(P \cap G[C_i]) \leq \frac{1}{\sqrt{n}}\card{C_i}\cdot\card{P \cap G[C_i]} \leq \frac{\card{C_i}^2}{\sqrt{n}}.\]

Let us define an SCC $C_i \in \cC$ to be \emph{large} if $|C_i| \geq 2\sqrt{n}/\eps$, and \emph{small} otherwise.
We will arrange for $P'$ to satisfy $P' \cap G[C_i] \subseteq P \cap G[C_i]$ for all small SCCs.
Then, by the inequality above, it is easy to check that
\[\sum_{\textrm{small } C_i} \pi(P' \cap G[C_i]) \le \sum_{\textrm{small } C_i} \pi(P \cap G[C_i]) \leq 2n/\eps = O(n).\]
The large $C_i$ are more troublesome, and this is the only place where we will use the star-edges contained in $H$.
The idea is that for every large $C_i$, if $u$ is the first vertex on $P$ in $C_i$ and $v$ is the last, then $P'$ reroutes from $u$ to $v$ using the star-edges of $C_i$ (more specifically, for a center $c \in C_i$, we replace the subpath of $P$ from $u$ to $v$ with the star-edges $uc$ and $cv$ and call the resulting path after performing all replacement operations $P'$).
Each of these star-edges has $\pi(e) \leq \card{C_i}/\sqrt{n}$, so the contribution by $\sum_{\textrm{large } C_i} \pi(P' \cap G[C_i])$ is negligible.
A potential worry is that $P'$ might have larger cost and delay than $P$.
But note that there are at most $\eps\sqrt{n}/2$ large SCCs, and using star-edges incurs an additive error to the cost and delay of at most $2D=2\sqrt{n}$ in each large SCC, so the total additive error on $P'$ caused by rerouting is $\leq \eps n$, which satisfies the Path-Sum Constraint.

\underline{Back-edges}:
By the Sparse Hitting guarantee of LDDs (\Cref{def:ldd}), every edge $e$ is in $\eback{}$ with probability $\Ot((\ell(e) +d(e))/\sqrt{n})$. 
Since $\ell(P) = d(P) = n$, we can use linearity of expectation to conclude $\expect{\card{P \cap \eback{}}} = \Ot(\sqrt{n})$.
Since every edge $e$ clearly has $\pi(e) = O(\sqrt{n})$ and $P' \cap \eback{} \subseteq P \cap \eback{}$ (note that by construction $P' \cap F \subseteq P \cap F$ for any $F \subseteq E(G)$), we have 
\begin{align*}
    \expect{\pi(P' \cap \eback{})} \le \expect{\pi(P \cap \eback{})} = \Ot(n).
\end{align*}

\underline{Forward-edges}: 
It is not hard to check that:
\begin{align*}
\pi(P' \cap \eforward{}) \le \pi(P' \cap \eback{}) +  \pi(P' \cap \textrm{intra-component edges}) + \Ot(\sqrt{n}).
\end{align*}
The reason is that beyond going from $\tau = 1$ to $\tau = n$ once, any movement of $P$ that goes forward along~$\pi$ must be counter-balanced by an earlier backward movement, and all such backward movement is due to either the back-edges or the intra-component edges.

Summing up over all the types of edges discussed above, we have that $\pi(P') = \Ot(n)$ in expectation.
Constructing multiple copies of $\pi, H$ with fresh randomness each time then results in the Path-Sum Constraint being satisfied with high probability.

\subsubsection{Dense Directed Graphs: Further Improvement (LDD hierarchy)}
We now briefly sketch how to achieve $\Pi = \Ot(n)$ instead of $\Ot(n^{3/2})$.
The idea is to apply LDD in a hierarchical fashion.
We first execute 
$$\eback{1} \leftarrow \LDD(\lendel{G}, D_1 = n/2)$$
and compute star-edges within each component.
Note that for the path $P$ assumed by the Path-Sum Constraint, we have $\expect{\card{P \cap \eback{1}}} = \Ot(1)$, so the total contribution of $\pi(P \cap \eback{1})$ will be small; the same applies to forward-edges at this level.
But the intra-component edges can now have a huge contribution, so we recursively compute an LDD within each component, now with $D=n/4$; we then add star-edges to the level $2$ components and continue recursing in this fashion.
In the end, we have a hierarchy of SCCs, where the SCCs at level $i$ are arranged in topological order and are each contained in some parent SCC at level $i-1$.
Such a hierarchy imposes a natural ordering $\tau$ on the vertices: $\tau$ is contiguous within every SCC and obeys the ordering at every level.
We then set $\pi(uv) = \card{\tau(v) - \tau(u)}$ (no down-scaling this time). 

The proof that $\pi$ obeys the Path-Sum Constraint is similar in spirit to \Cref{sub:overview-dense-slow}, though we need to be more careful in analyzing relationships between levels.
There is also one additional difficulty worth noting.
Recall that $P$ is the ``optimal" path assumed by the Path-Sum Constraint, and say that there is a component $C$ at some level $i$ such that $\ell(P \cap G[C]) + d(P \cap G[C])$ is large relative to $D_i$; we call such SCCs \emph{long} for $P$.
In this case, $\eback{i}$ will in expectation contain many edges from $P \cap G[C]$, which prevents us from upper-bounding $\pi(P \cap G[C])$.
Fortunately, we can show that long SCCs are rare, and hence they too can be rerouted along star-edges.
Even though we handle them the same way, there is a key conceptual difference between large and long SCCs.
Largeness depends only on $\card{C}$, so the algorithm can detect large SCCs when constructing $\pi, H$; in fact, instead of adding star-edges to large SCCs, we could have the algorithm contract them instead.
But the algorithm \emph{cannot} detect long SCCs, because it has no knowledge of $P$; here contraction would not work, and the crucial advantage of adding edges to $H$ is that it allows the algorithm to (non-constructively) guarantee the existence of the desired $P'$ by rerouting in the analysis only.

\subsection{Computing $\pi, H$ in Sparse Graphs}

\subsubsection{A Different Starting Point: Sparse DAGs}

We start by constructing $\pi, H$ for a sparse DAG.
This algorithm uses our algorithm for All-Pairs $\rsp$ (see \Cref{thm:all-pairs-simple}) as a black-box; this latter algorithm uses a different and simpler set of techniques (though it still uses the $\pidp{\one}$ as a key ingredient), so we do not discuss its inner machinery in this overview. 

\paragraph{Hop-Edges.}
We introduce an additional type of edges to $H$, the \emph{hop-edges}.
Say that we compute All-Pairs $\rsp$ in $G[S]$ for some $S \subset V$.
Then, for every $u,v \in S$, we have the option to add an edge $uv$ to $H$ whose length and delay are within a $(1+\eps)$-factor of the \textit{shortest delay-constratined} $(u,v)$-path in $G$; there can be many length/delay tradeoffs from $u$ to $v$ in $G$, so we add $\Ot(1)$ many parallel $uv$ edges to~$H$ that represent all possible tradeoffs (up to a $(1+\eps)$-factor).
Note that whereas star-edges only contain a loose upper bound on cost/delay, the hop-edges are accurate because they are based on our computation of all-pairs restricted distances; so shortcutting a path~$P$ using hop edges does not incur any significant error.

\paragraph{Constructing $\pi, H$.}
Below is an outline of our construction for $\pi, H$ in a sparse DAG.
\begin{enumerate}
    \item Compute a topological ordering $\tau$ and partition the vertices into $n/\blocksize$ blocks $B_i$ of size $\blocksize$, where each block takes the form $\{v : \tau(v) \in [x, x + \blocksize)\}$ for some $x$.
    \item Compute All-Pairs $\rsp$ within each $G[B_i]$ and add corresponding hop edges $uv$ for some pairs $u,v \in B_i$
    \item These hop-edges will ensure that every path $P$ can be (approximately) shortcut to contain $\ll n$ edges, which will allow us to increase $\pi$-values while still satisfying the Path-Sum Constraint, which in turn reduces $\Pi$ and speeds up $\pidp{\pi}$ on $G \cup H$.
\end{enumerate}

We now elaborate on steps 2 and 3 of the outline.
For the sake of intuition, let us assume $m = \Ot(n)$.
Since we compute All-Pairs $\rsp$ inside each $G[B_i]$, it is tempting to add hop edges for every pair $u,v \in B_i$.
The problem here is that we would then have $\card{H} \gg \card{E(G)}$, which would cause $\pidp{\pi}$ on $G \cup H$ to run slow.

We instead only add hop-edges to some vertex-pairs in each $B_i$.
In each $B_i$, we sample every vertex into a set $S_{B_i}$ with probability $\Theta(\log(n)/\blockrun)$, for some parameter $\blockrun < \blocksize$, and we add hop edges $uv$ to $H$ for every $u,v \in S_{B_i}$.
As a result, for any path $P$, the segment of $P \cap G[B_i]$ can be shortcut to contain at most $O(\blockrun)$ edges: with high probability the first and last $\blockrun$ edges of $P$ contain a vertex in $S_{B_i}$, and we can take a hop-edge from $H$ between those two vertices in $S_{B_i}$.
Since $\card{S_{B_i}} = \Ot(\blocksize / \blockrun)$, we have
\begin{equation}
\label{eq:overview-num-edges}
    \card{H} = \Ot\paren{\frac{n}{\blocksize} \cdot \paren{\frac{\blocksize}{\blockrun}}^2} = \Ot\paren{n \frac{\blocksize}{\blockrun^2}}
\end{equation}
and
\begin{equation}
\label{eq:overview-hop-diam}
    \textrm{every} (s,t)\textrm{-path } P \textrm{ can be shortcut to contain roughly at most } n \frac{\blockrun}{\blocksize} \textrm{edges}.
\end{equation}

Since we are assuming $\card{E(G)} = \Ot(n)$, a natural approach is to optimize between $\blocksize$ and $\blockrun$ while always ensuring that $\card{H} = \Ot(n)$.
This results in $\blocksize = n^{2/3}$ and $\blockrun = n^{1/3}$.
By \Cref{eq:overview-hop-diam}, every path can be shortcut down to $\Ot(n^{2/3})$ edges, so we can set  $\pi(e) = n^{1/3}$ for all edges $e$ in $G \cup H$, while still satisfying the Path-Sum Constraint.
This yields $\Pi = m/n^{1/3}$, so by \Cref{thm:pi-dp}, the overall runtime is $mn^{2/3}$.

The flexibility of using a $\pi$ function allows us to attain further improvement of \Cref{thm:sparse-simple}.
Somewhat counter-intuitively, even if $m = \card{E(G)} = \Ot(n)$, we will allow $H$ to contain $n^{1+\delta}$ edges (for some small constant $\delta$).
We can afford this because the edges of $H$ make things easier to handle in the following way: for any path $P$ and block $B_i$, $P$ can be shortcutted to use $\blockrun$ edges in $G[B_i]$ but \emph{only one} edge in $G[B_i] \cap H$.
Decreasing $\blockrun$ allows us to set $\pi(e)$ to be larger for edges in $G[B_i]$ while still satisfying the Path-Sum Constraint.
Therefore, while we may add more edges into~$H$, we have room to increase the $\pi$-value of each edge from $G[B_i]$ which results in a smaller $\Pi$. 

\subsubsection{Sparse Directed Graphs}
We again use an LDD Hierarchy to extend the DAG-approach to general graphs.
This is the most technically involved part of our paper, so we only give a very brief sketch here.
Recall that for dense graphs, although the \emph{analysis} distinguished between different kinds of SCCs (e.g. long or large), the algorithm itself treated every SCC the same way: add star-edges and recurse.
For the sparse case, by contrast, the construction will treat SCCs differently depending on their size.

We will set parameters $\blocksize$ and $\blockrun$ the same as for DAGs.
At the top level of the hierarchy, we compute:
$$\eback{1} \leftarrow \LDD(\lendel{G},D = n/2)$$
Let us say that an SCC is small if it has size less than $\blocksize$.
We end up denoting non-small SCCs as ``medium" SCCs (there are also large SCCs, but these are easy to handle via star-edges, just like in the dense case).
The basic idea is to group small SCCs into blocks and then run the DAG-approach on each block (i.e. compute All-Pairs $\rsp$ and add hop edges).
Medium SCCs cannot be handled this way because we cannot afford to compute All-Pairs $\rsp$ in a large set of vertices; we thus recursively call LDD on the medium SCCs. 

The resulting LDD hierarchy thus has different kinds on SCCs at every level, and so requires a more delicate analysis of the interaction between levels, which we leave for the main body of the paper.
One immediate difference with the dense case is that in the former the only edges of $H$ are star-edges, which are trivial to analyze and have minimal contribution to $\pi$; by contrast, we need to be more careful in bounding the contribution of the hop-edges.
Recall that in the sparse DAG analysis, we claimed that for any block $B_i$, $P \cap G[B_i]$ only uses a single edge from $H$; this is no longer true in general graphs because the path can go in and out of $B_i$, and each such segment might use an edge from $H$.
To overcome this, we observe that $P$ can only go in and out of $B_i$ by using back-edges somewhere in the LDD hierarchy, and we can use the structure of the LDD hierarchy to upper bound the number of such back-edges.

\section{General Framework For Directed Graphs}
\label{sec:framework}

In this section, we provide a setup that is shared by both our main results.
We first reduce the problem of finding a $(1+\eps,1+\eps)$-approximation of $\rsp$ on general graphs to solving a decision version of the problem, Gap $\rsp$.
Then, we give a setup for solving Gap $\rsp$ by using $\pidp{\pi}$ for a well-chosen $\pi$.

\subsection{Reduction to Gap $\rsp$}

Consider the following decision problem Gap $\rsp$ (note the asymmetry in the $\mathtt{NO}$ case).

\begin{problem}[Gap $\rsp$]
    Given $s \in V$ and $L,D \in \IR_{\ge 0}$ and $\eps>0$, answer for all $t$:
    \begin{itemize}
        \item $\mathtt{YES}$ if there is an $(s,t)$-path with length $\le L$ and delay $\le D$.
        \item $\mathtt{NO}$ if there is no $(s,t)$-path with length $\le (1 + \apx)L$ and delay $\le (1 + \apx)^2 D$.
        \item $\mathtt{YES}$ or $\mathtt{NO}$ otherwise (any output is fine).
    \end{itemize}
\end{problem}

The following lemma formally states our reduction from $\rsp$ to Gap $\rsp$; the proof is a standard boosting argument and is deferred to \Cref{sec:reduction-to-gap}.

\begin{restatable}{lemma}{reduction}
\label{lem:reduction-to-gap}
    Let $\alpha$ be an arbitrary constant and let $\cA$ be an $O\paren{T(m,n) / \eps^\alpha}$ time randomized algorithm that solves Gap $\rsp$ with the following one-sided error: for all $t \in V$, $\cA$ answers correctly with probability at least $1/2$ in the $\mathtt{YES}$ case but $\cA$ must always answer correctly in the $\mathtt{NO}$ case.
    Then, there is a randomized algorithm $\cA^*$ which gives a $(1 + \eps, 1 + \eps)$ approximation to $\rsp$ with high probability in time $O\paren{(T(m,n) + n) \log^{\alpha + 2}(n) \log(nW) / \eps^{\alpha + 1}}$.
\end{restatable}

We may thus turn our attention to solving Gap $\rsp$.

\subsection{Setup for Solving Gap $\rsp$}

Both the algorithms behind \Cref{thm:dense-simple,thm:sparse-simple} (details in \Cref{sec:dense,sec:sparse} respectively) share a common high-level framework for solving Gap $\rsp$ which we outline here.

\subsubsection{Normalization of Lengths and Delays}
The approach we take to computing Gap $\rsp$ involves the following renormalization of lengths and delays.
We set $\hatell(e) = n \cdot \ell(e)/(\eps \cdot L)$ and $\hatd(e) =n \cdot d(e)/(\eps \cdot D)$.
Notice that for all paths $P$,
\begin{align*}
    \ell(P) \le L
    \textrm{ and }
    d(P) \le D
    \implies
    \hatell(P) \le \frac{n}{\eps}
    \textrm{ and }
    \hatd(P) \le \frac{n}{\eps}
\end{align*}
and
\begin{align*}
    \hatell(P) > (1 + \apx)\frac{n}{\eps}
    \textrm{ or }
    \hatd(P) > (1 + \apx)^2\frac{n}{\eps}\\
    \implies
    \ell(P) > (1 + \apx) L
    \textrm{ or }
    d(P) > (1 + \apx)^2 D.
\end{align*}
This suggests (we are being rather informal here) that a suitable test is to answer $\mathtt{YES}$ if and only if there is a path with normalized length and delay each at most $n/\eps + O(n \log n)$.
Roughly speaking, we conduct this test by computing a frequency function $\pi$ such that if there exists an $(s,t)$-path $P$ with $\hatell(P) \le n / \eps$ and $\hatd(P) \le n / \eps$, then there exists an $(s,t)$-path $P'$ with $\hatell(P') \le (1 + \apx) n/\eps$ and $\hatd(P') \le (1 + \apx) n/\eps$ and, lastly, that satisfies $\pi(P') = \Ot(n/\eps)$; running $\pidp{\pi}$ up to depth $h = \Ot(n/\eps)$ with delay threshold $(1 + \apx)n/\eps$ and approximation parameter $\apx$ would then complete the test.
This is, however, just a rough outline.
In fact, we do not run $\pidp{\pi}$ on the original input graph $G$; we run it on $G \cup H$ for a random edge set $H$.
More on this now.

\subsubsection{Auxiliary Edges}
Our algorithms for Gap $\rsp$ will first prepare a random set $H$ of edges and then run $\pidp{\pi}$ on $G \cup H$.
In order for this strategy to be feasible, we will need to take care of a few things.

\begin{definition}[Pareto Frontier Preserving Auxiliary Edges]
\label{def:good-edges}
    An edge set $H$ is said to be \emph{Pareto Frontier Preserving} with respect to $G$ if for all edges $uv \in H$, there is a $(u,v)$-path $P_{uv} \subseteq G$ such that $\ell(P_{uv}) \le \ell(uv)$ and $d(P_{uv}) \le d(uv)$.
\end{definition}
This property ensures that any valid solution to Gap $\rsp$ in $G \cup H$ also corresponds to a valid solution in $G$.
Moreover,~as we are running $\pidp{\pi}$ on $G \cup H$, we must accordingly define~$\pi$ not only for $E(G)$ but also for $H$.

For a set of edges $H$ that is Pareto Frontier Preserving, any $(s,t)$-path in $G \cup H$ must be (weakly) Pareto dominated by an $(s,t)$-path in $G$; call the latter path $P$.
If $\hatell(P) \le n/\eps$ and $\hatd(P) \le n/\eps$, then the desired approximate path $P'$ mentioned earlier can be comprised of edges from $G \cup H$; it is in fact $H$ which gives us the latitude to efficiently define $\pi$ in such a way that $\pi(P') = \Ot(n/\eps)$.

\subsubsection{Randomness of the Approximate Path $P'$}
Now let us take randomness into account (randomness of both $H$ and $\pi$).
Earlier we stated the desideratum that if there exists an $(s,t)$-path $P$ in $G$ with $\hatell(P) \le n / \eps$ and ${\hatd(P) \le n / \eps}$, then there must exist an $(s,t)$-path $P'$ in $G \cup H$ with $\hatell(P') \le (1 + \apx) n/\eps$ and ${\hatd(P') \le {(1 + \apx) n/\eps}}$ and, lastly, that satisfies $\pi(P') = \Ot(n/\eps)$.
In reality, the existence of such a $P'$ is not deterministically guaranteed: instead, we will show a random process to construct an $(s,t)$-path~$P'$ from $P$ such that $\hatell(P') \le (1 + \hapx) n/\eps$ and $\hatd(P') \le (1 + \hapx) n/\eps$ and, critically, that satisfies $\expect{\pi(P')} = \Ot(n/\eps)$.
By Markov's inequality, $\pi(P') = \Ot(n/\eps)$
happens with probability at least $1/2$, which satisfies the hypothesis of \Cref{lem:reduction-to-gap}.
Note that the random process for constructing $P'$ is solely for the purpose of analysis and our algorithms for Gap $\rsp$ never explicitly construct such paths.
We capture that this property holds by the following definition.
\begin{definition}[Path-Sum Constraints]
\label{def:path-sum-constraints}
    We say $\pi, H$ $\alpha$-satisfy the \emph{Path-Sum Constraint} for $t \in V$ if, when there exists an $(s,t)$-path $P \subseteq G$ with $\hatell(P) \le n/\eps$ and $\hatd(P) \le n/\eps$, over the randomness of $\pi, H$ there is an $(s,t)$-path $P' \subseteq G \cup H$ such that $\hatell(P') \le (1 + \apx)n/\eps$ and $\hatd(P') \le (1 + \apx)n/\eps$ and $\expect{\pi(P')} = O(n \log^{\alpha}(n) / \eps)$ for some constant $\alpha$.
\end{definition}
When there is no need to be precise about $\alpha$ in \Cref{def:path-sum-constraints}, we just say that the Path-Sum Constraints are satisfied (as opposed to $\alpha$-satisfied).

\paragraph{On Integral Valued $\pi$. }
The inner-workings of $\pidp{\pi}$ very much depend on the range of $\pi$ taking on positive integer values, which a careful reader of the later sections may notice we sweep under the rug.
We can remedy this by composing the ceiling function with $\pi$.
This decreases $\Pi$ (which is in our favor), and possibly increments $\pi(P')$ by $n$ since $P'$ is simple and thus uses at most $n$ edges.
The increment to $\pi(P')$ by $n$ still maintains the Path-Sum Constraints, which require that $\pi(P') = \Ot(n/\eps)$.

\subsubsection{Summary}

All of the tools in this section allow us to reduce the problem of solving Gap $\rsp$ to the problem of finding $\pi,H$ that satisfy the Path-Sum Constraints for all $t \in V$ and have small $\Pi$ (recall from the statement of \Cref{thm:pi-dp} that $\Pi = \sum_e 1/\pi(e)$); the key problem is thus to compute such $\pi, H$, which we do for \Cref{thm:dense-simple,thm:sparse-simple} (in \Cref{sec:dense,sec:sparse} respectively).
Formally, we have the following reduction.

\begin{lemma}
\label{lem:gap-solver}
    If we can compute in time $T$ both a Pareto Frontier Preserving set of edges $H$ and a $\pi$ such that $\pi, H$ $\alpha$-satisfy the Path-Sum Constraints for all $t \in V$, then we can solve Gap $\rsp$ in time $O\paren{T + n \Pi \log^{\alpha}(n) / \eps^2}$ up to one-sided error: for all $t \in V$, we answer~correctly with probability at least $1/2$ in the $\mathtt{YES}$ case but must always answer correctly in the $\mathtt{NO}$ case.
\end{lemma}
\begin{proof}
    We use the following algorithm to solve Gap $\rsp$.
    \begin{tbox}
    \underline{\textbf{\textsc{Gap $\rsp$ Solver}}}
        \begin{enumerate}
            \item $\hatd(e) \gets \frac{d(e)}{\eps D / n}$ for all $e \in E$.
            \item $\hatell(e) \gets \frac{\ell(e)}{\eps L / n}$ for all $e \in E$.
            \item $\hat{G} \gets (V, E, \hatell, \hatd)$.
            \item Compute $H$, a Pareto Frontier Preserving set of edges.
            \item Compute $\pi$ such that $\pi, H$ $\alpha$-satisfy the Path-Sum Constraints for all $t \in V$.
            \item Run $\pidp{\pi}$ on $\hat{G} \cup H$ up to depth $h = \Theta(n \log^{\alpha}(n) / \eps)$ with source $s$, delay threshold $(1 + \apx)\frac{n}{\eps}$, and approximation parameter $\apx$.
            \item If $\pidp{\pi}$ answers $\le \frac{n}{\eps} + O(n \log n)$ for $t \in V$, say $\mathtt{YES}$ for $t$; otherwise say $\mathtt{NO}$ for $t$.
        \end{enumerate}
    \end{tbox}
    The running time of the algorithm above is immediate from the running time stated in \Cref{thm:pi-dp} and noting that we have set the depth $h = \Theta(n \log^{\alpha}(n) / \eps)$ and used the approximation parameter $O(\eps \log n)$ on the call to $\pidp{\pi}$.
    We next show correctness.

    \underline{Completeness (via Path-Sum Constraints being satisfied)}:
    For any $t \in V$, suppose there is an $(s,t)$-path $P$ in $G$ with $\ell(P) \le L$ and $d(P) \le D$.
    Then, $\hatell(P) \le n / \eps$ and $\hatd(P) \le n / \eps$.
    Since the Path-Sum Constraints are $\alpha$-satisfied for $t$, with probability at least $1/2$ there is an $(s,t)$-path $P'$ using edges from $G$ and $H$ with $\hatell(P') \le (1 + \apx) n / \eps$ and $\hatd(P) \le (1 + \apx) n / \eps$ and $\pi(P') = O(n \log^{\alpha}(n) / \eps)$.
    By \Cref{thm:pi-dp}, the call to $\pidp{\pi}$ on $\hat{G} \cup H$ produces an answer which is at most $(1 + \apx) n / \eps$ for $t$, and hence the algorithm answers $\mathtt{YES}$.

    \underline{Soundness (via Pareto Frontier Preservation)}:
    For any $t \in V$, suppose there is no $(s,t)$-path $P$ in $G$ with $\ell(P) \le (1 + \apx) L$ and $d(P) \le (1 + \apx)^2 D$.
    Then, there is no $(s,t)$-path $P$ with $\hatell(P) \le (1 + \apx) n / \eps$ and $\hatd(P) \le (1 + \apx)^2 n / \eps$.
    Since $H$ is Pareto Frontier Preserving, there will be no $(s,t)$-path $P$ in $\hat{G} \cup H$ with $\hatell(P) \le (1 + \apx) n / \eps$ and $\hatd(P) \le (1 + \apx)^2 n / \eps$.
    By \Cref{thm:pi-dp}, the call to $\pidp{\pi}$ on $\hat{G} \cup H$ must produce an answer which is greater than $(1 + \apx) n / \eps$ for $t$, and hence the algorithm answers $\mathtt{NO}$.
\end{proof}

\section{$\rsp$ on Dense Directed Graphs}
\label{sec:dense}

In this section, we complete \Cref{thm:dense-simple}.
Recall that, in view of \Cref{lem:gap-solver}, our goal is to efficiently construct an auxiliary edge set $H$ that is \PFP and a frequency function $\pi$ so that $\pi, H$ satisfy the \PSCs for all $t \in V$.

Our construction of $\pi, H$ heavily relies on directed LDDs from \Cref{def:ldd}.
Note, however, that our input graph has two weight functions (length and delay), while the LDD from \Cref{def:ldd} only handles a single weight function. 
Fortunately, we show that it is sufficient for us to use the regular (single-weight) LDD, by using a weight function that combines cost and delay into a single number.

\comblendel*
Throughout this section, we let $n = \card{V(\Gin)}$, where $\Gin$ is the initial input graph.

\subsection{Frequencies and Auxiliary Edges via a Directed LDD Hierarchy}
\begin{figure}[b!]
    \centering
    \includegraphics[scale=0.7]{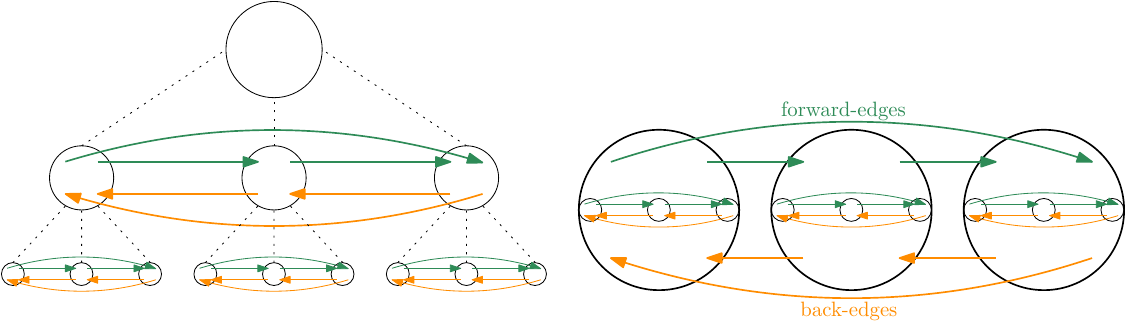}
    \caption{
        Diagram of a $2$ level LDD Hierarchy, with the root in level $0$ (for convenience), and leaves in level $2$.
        Each node in levels $i > 0$ are SCCs contained in their parent.
        On the right, the same LDD Hierarchy, flattened.
        Star-edges have not been depicted.
    }
    \label{fig:hierarchy-no-star}
\end{figure}

Intuitively, our function $\pi$ will aim to mimic the topological ordering from Example 2 of \Cref{subsec:examples}, but now applied to general (not acyclic) graphs.
We use an LDD Hierarchy with $\log n$ levels with the following data:
\begin{itemize}
\item In level $1$, 
(i) a Directed LDD of $\Gin$ (Recall from \Cref{def:ldd} that the LDD is a set of edges $B$ which here we call ``back-edges''); 
(ii) the SCCs of $\Gin \setminus B$; 
(iii) the edges of the DAG obtained from $\Gin \setminus B$ by contracting all SCCs, which we call ``forward edges''.
\item And in level $i > 1$, for each SCC $C$ in level $i - 1$, 
(i) a Directed LDD of $G[C]$; 
(ii) the SCCs of $G[C]$ after removing back-edges of this LDD, each with a pointer to $C$ in the LDD Hierarchy; 
(iii) the forward edges of this LDD (defined analogously as above).
\end{itemize}
For any two SCCs $C$ and $C'$ of the LDD Hierarchy, we say that $C'$ is an ancestor of $C$ if $C \subseteq C'$.
\Cref{fig:hierarchy-no-star} depicts the aggregate of data in the LDD Hierarchy described so far.

To construct the LDD Hierarchy, we call the following simple algorithm on $\Gin$ with recursion depth $i = 1$ and interval $[n]$. (Details on the interval, $\tau$, and star-edges to follow shortly; it is safe to ignore them (the grayed-out text) for now.)
\begin{tbox}
    \algname{Low-Diameter Decomposition Hierarchy (Dense)}\\
    \textbf{Input:} A directed graph $G = (V, E, \ell, d)$, recursion depth $i$\textcolor{gray}{, and an interval~$I$}.\\
    \textbf{Output:} Data of an LDD Hierarchy.
    \begin{enumerate}
        \item If $i < \log n$, then: \textcolor{blue}{\texttt{(defining level $i$ data)}}
            \begin{enumerate}
                \item $D_i \gets n/2^i$.
                \item Add $\LDD(\lendel{G}, D_i)$ to $\eback{i}$.
                    \textcolor{blue}{\texttt{(These are back-edges)}}
                \item $\cC = \set{C_1, C_2, \ldots, C_k} \gets \SCC(G \setminus \eback{i})$.
                \item \textcolor{gray}{Add star-edges with length and delay $D_i$ into $\estar{i}$ (which is disjoint from $E(G)$).}
                \item \textcolor{gray}{$I = I_1 \sqcup I_2 \sqcup \ldots \sqcup I_k$ where $I_j$ are intervals increasing in $j$ with $\card{I_j} = \card{C_k}$.}
                \item \label{alg:dense-tau} \textcolor{gray}{Update $\tau$ so that $\tau(C_j) = I_j$ for all $j \in [k]$.}
                \item Add $(G \setminus \eback{i}) \setminus \bigcup_{C \in \cC} E(G[C])$ to $\eforward{i}$.
                    \textcolor{blue}{\texttt{(These are forward-edges)}}
                \item For all $j \in [k]$, recurse into $G[C_j]$ with recursion depth $i+1$ \textcolor{gray}{ and interval $I_j$}.
            \end{enumerate}
    \end{enumerate}
\end{tbox}
We are almost ready to define $\pi$, but first let us define the auxiliary edge set $H$.

\paragraph{Auxiliary Edges. }
For each SCC $C$ stored in the level $i$ data of the LDD Hierarchy, we create a bidirected star by selecting an arbitrary vertex $c \in C$ and, for all $v \in C$, adding the edges $cv$ and $vc$ into $\estar{i}$ with $\ell(cv) = \ell(vc) = d(cv) = d(vc) = D_i$; we then let $H = \bigcup_i \estar{i}$.
We call the edges of~$H$ ``star-edges''.
Note that we treat $H$ as disjoint from $E(G)$ so from here on, even though~$G$ is simple, $G \cup H$ may have parallel edges.

\begin{observation}
\label{obs:dense-H-pfp}
    $H$ is \PFP.
\end{observation}
\begin{proof}
    This follows immediately since for any $i$, for any $uv \in \estar{i}$, $u$ and $v$ are in an SCC at level~$i$ and so there is a path from $u$ to $v$ with combined length-delay $D_i$ (by \Cref{def:ldd}); this path Pareto dominates the edge $uv \in H$, which has $\ell(uv) = d(uv) = D_i$.
\end{proof}

\paragraph{Back to $\pi$. }
The following definition gives an analogue of a topological order on DAGs to LDD Hierarchies.

\begin{definition}[LDD Hierarchy Respecting Order]
\label{def:order}
    We say that an ordering $\tau: V \rightarrow [n]$ respects an LDD Hierarchy if it is a bijection satisfying the following:
    \begin{itemize}
        \item For any SCC $C$ in any level of the LDD Hierarchy, the image of $C$ under $\tau$ is an interval $I \subseteq [n]$.
        \item For any pair of SCCs $C, C'$ in level $i$ of the LDD Hierarchy such that $C$ precedes $C'$ in a topological order of the SCCs of level $i$, the interval $\tau(C)$ precedes $\tau(C')$.
        \item For any pair of SCCs $C, C'$ of the LDD Hierarchy such that $C \subseteq C'$, we have $\tau(C) \subseteq \tau(C')$.
    \end{itemize}
\end{definition}
Observe that, by construction, $\tau$ produced by the LDD Hierarchy above in \hyperref[alg:dense-tau]{Step (f)} satisfies \Cref{def:order}.

We may then define $\pi$ like so:
\begin{align*}
    \pi(uv) = \card{\tau(u) - \tau(v)}.
\end{align*}

We can immediately make the following observation (see \Cref{thm:pi-dp} to recall the meaning of~$\Pi$).
\begin{observation}
\label{obs:dense-Pi}
    $\Pi = O(n \log n)$.
\end{observation}
The proof is exactly the same as in Example 2 of \Cref{subsec:examples}, but we bound the contribution of $E(G)$ and $H$ separately\footnote{See \Cref{sec:parallel-edges} to see what changes when $\Gin$ has parallel edges.}.

\subsection{Verifying the Path-Sum Constraints}
Fix an arbitrary $(s,t)$-path $P$ in $G$ with $\ell(P) \le n/\eps$ and $d(P) \le n/\eps$ for the remainder of this section, where our aim is to exhibit an $(s,t)$-path $P'$ in $G \cup H$ which witnesses that the \PSC for $t$ is satisfied.
Note that although our description of $P'$ is constructive, the algorithm itself never computes $P'$; we simply need to show the existence of such a $P'$ to verify the \PSCs.
First, a couple of definitions which we use to construct $P'$.

\begin{definition}[Large SCCs]
\label{def:dense-large}
    An SCC $C$ at level $i$ of the LDD Hierarchy is said to be \emph{large} if $\card{C} \ge D_i = n/2^i$.
\end{definition}

\begin{definition}[Long and Short SCCs]
\label{def:dense-long}
    Let $D_i = n/2^i$.
    An SCC $C$ at level $i$ of the LDD Hierarchy is said to be \emph{long} for $P$ if $\ell(P \cap G[C]) > D_i/\eps$ or $d(P \cap G[C]) > D_i/\eps$.
    Otherwise, we say that $C$ is \emph{short} for $P$.
\end{definition}

Large and long SCCs are ``troublesome'' in that they may contribute too much to $\pi(P)$ (hence why we look to $\pi(P')$ instead).
For example, if $C$ is a large SCC at level $\log n - 1$ (the bottom level), then the ``average" edge $(u,v) \in E(C)$ has $\pi(u,v) = |\tau(v) - \tau(u)| \sim |C|/2$, so if $\card{P \cap G[C]} = \Omega(\card{C})$, then $\pi(P \cap G[C])$ could be as much as $\Theta(\card{C}^2)$, which is too much when $\card{C}$ is large.
Similarly, the number of back-edges in $P \cap G[C]$ controls the magnitude of $\pi(P \cap G[C])$ and if $C$ at any level of the LDD Hierarchy were to be long for $P$, the number of back edges would be very large in expectation.
Ideally, we can contract large and long SCCs to remove their contribution to $\pi(P)$ and pay an error in length and delay for such a contraction.
This is in fact easy to do for large SCCs.
However, an SCC may be simultaneously long for one path and short for another and we lack the foresight as to what $P$ is.
We instead simulate a contraction by replacing the part of a path going through a long SCC with up to two star-edges.

Intuitively, the path $P'$ that we construct is the same as the path $P$, except that we shortcut all the long SCCs with star-edges (simulating a contraction of said SCC), as well as large SCCs at the bottom level.
We now describe the construction process more formally.
\clearpage
\begin{tbox}
    \algname{Construction of $P'$ Given $P$}\\
    \textbf{Input:} An $(s,t)$-path $P$ with $\ell(P) \le n/\eps$ and $d(P) \le n/\eps$.\\
    \textbf{Output:} An $(s,t)$-path $P'$ which witnesses that the \PSC for $t$ is satisfied.
    \begin{enumerate}
        \item $P' \gets P$.
        \item For $i \in [\log n - 1]$ in increasing order:
            \begin{enumerate}
                \item $\cC \gets $ SCCs of level $i$ of the LDD Hierarchy (i.e. the SCCs of $G \setminus \paren{\bigcup_{j \le i} \eback{j}}$).
                \item For each SCC $C \in \cC$
                    that does not have an ancestor SCC that is long for $P$:
                    \begin{enumerate}
                        \item Let $c$ be the representative of $C$ (i.e. the center of the star).
                        \item Let $Q$ be the subpath from $u$ to $v$, where (if they exist) $u$ is the first vertex of $P'$ in $C$ and $v$ is the last vertex of $P'$ in $C$.
                        \item If $C$ is long for $P$, then $P' \gets (P' \setminus Q) \cup \set{uc, cv}$.
                        \item Else if $i = \log n - 1$ and $C$ is large, then $P' \gets (P' \setminus Q) \cup \set{uc, cv}$.
                    \end{enumerate}
            \end{enumerate}
        \item Output $P'$
    \end{enumerate}
\end{tbox}

Note that for any SCC $C$ in the LDD Hierarchy, $G[C] \cap P' \subseteq G[C] \cap P$.
Furthermore, the next two observations will come in useful for formally bounding $\expect{\pi(P')}$.

\begin{observation}
\label{obs:dense-has-back-then-short}
    Let $C$ be an SCC at level $i - 1$ of the LDD Hierarchy, and let ${X = \eforward{i} \cup \eback{i}}$ (i.e. the forward-edges and back-edges in level $i$).
    If $G[C] \cap P' \cap X \neq \emptyset$, then $C$ is short for $P$.
\end{observation}
\begin{proof}
    We prove the contrapositive of the statement.
    If $C$ was long for $P$, then $G[C] \cap P'$ would be empty or comprised of star-edges, by construction; that is, $G[C] \cap P' \subseteq H$.
    Since $X \subseteq G$ is disjoint from~$H$, we must have $G[C] \cap P' \cap X = \emptyset$.
\end{proof}

\begin{observation}
\label{obs:dense-ldd-edges}
    Let $C$ be an SCC in level $i - 1$ of the LDD Hierarchy that is short for $P$, or $C = V$.
    \[\expect{\card{G[C] \cap \eback{i} \cap P'}} = O(\log^3(n) / \eps).\]
\end{observation}
\begin{proof}
    \begin{align*}
        \expect{\card{G[C] \cap \eback{i} \cap P'}}
        & \le
        \expect{\card{G[C] \cap \eback{i} \cap P}}
        \tag{$G[C] \cap \eback{i} \cap P' \subseteq G[C] \cap \eback{i} \cap P$}
        \\
        & =
        \sum_{e \in G[C] \cap \eback{i} \cap P} O\paren{\frac{w(e)}{D_i}\log^3 n + \frac{1}{\poly(n)}}
        \tag{\Cref{prop:ldd} on level $i$ of the LDD Hierarchy}
        \\
        & =
        O\paren{\log^3(n) / \eps}.
        \tag{Linearity of Expectation, and $C$ is short for $P$}
    \end{align*}
\end{proof}
We are now ready to bound $\expect{\pi(P')}$. Let us start with the following observation.
\begin{observation}
\label{obs:dense-forward}
    The total over all $i$ that $\eforward{i}$ contributes to $\pi(P')$ is at most $(\pi(P') + n)/2$.
\end{observation}
\begin{proof}
    Consider the sums
    \begin{align*}
        X_{+} & = \paren{\sum_{\substack{uv \in P':\\\tau(v) - \tau(u) > 0}} \tau(v) - \tau(u)},
        \\
        X_{-} &= \paren{\sum_{\substack{uv \in P':\\\tau(v) - \tau(u) < 0}} \tau(v) - \tau(u)}.
    \end{align*}
    By telescoping, $\card{X_{+} + X_{-}}  = \card{\tau(s) - \tau(t)}$ and so we have $X_{+} + X_{-} \le n$.
    Note that \linebreak ${\pi(P') = X_{+} - X_{-} \ge 2X_{+} - n}$, and consequently $X_{+} \le (\pi(P') + n)/2$.
    Since $\eforward{i}$ can only contribute to the $X_{+}$ term, the conclusion follows.
\end{proof}

\begin{lemma}
\label{lem:dense-pi}
    $\expect{\pi(P')} = O(n \log^4(n) / \eps)$.
\end{lemma}
\begin{proof}
    It suffices to bound the expected contribution to $\pi(P')$ by only the non-forward edges since the contribution by forward-edges to $\pi(P')$ is at most $(\pi(P') + n)/2$ by \Cref{obs:dense-forward}. As a result, $\expect{\pi(P')}$ is at most $n$ plus twice the expected contribution of non-forward edges to $\pi(P')$.
    
    That means we will show bounds for the contribution to $\expect{\pi(P')}$ by back-edges, star-edges, and (unlabelled) edges contained in the SCCs at level $\log n - 1$ of the LDD Hierarchy.
    Throughout, we use the fact that for any SCC $C$ in the LDD Hierarchy and any $u,v \in C$, we have $\pi(uv) = \card{\tau(u) - \tau(v)} \le \card{C}$, which follows from $\tau$ satisfying \Cref{def:order}.

    \underline{Back-edges}:
    Let $\cC$ be the SCCs of level $i - 1$ of the LDD Hierarchy.
    The total contribution to $\expect{\pi(P')}$ by back-edges $\eback{i}$ at level $i$ is at most
    \begin{align*}
        \sum_{uv\in \eback{i}} \pi(uv)\cdot \prob{uv\in P'}
        & = \sum_{C\in\cC}\sum_{uv\in \eback{i}:u,v\in C} (\tau(u)-\tau(v))\cdot \prob{uv\in P'}\\
        & \le \sum_{C \in \cC} \card{C} \cdot \expect{\card{G[C] \cap \eback{i} \cap P'}}\\
        & \le
        \sum_{\substack{C \in \cC:\\ C \textrm{ is short for } P}} \card{C} \cdot \expect{\card{G[C] \cap \eback{i} \cap P}}
        \tag{$G[C] \cap \eback{i} \cap P' \subseteq G[C] \cap \eback{i} \cap P$ and \Cref{obs:dense-has-back-then-short}}
        \\
        & =
        O(\log^3(n) / \eps) \sum_{C \in \cC} \card{C}
        \tag{\Cref{obs:dense-ldd-edges}}
        \\
        & =
        O(n \log^3(n) / \eps).
        \tag{SCCs are vertex disjoint}
    \end{align*}
    Summing over all $\log n$ levels gives an upper bound of $O(n \log^4(n) / \eps)$ on the expected contribution
    of back edges to $\pi(P')$.
    
    \underline{Star-edges}:
    Now let $\cC$ be the SCCs of level $i$ of the LDD Hierarchy.
    The total contribution to $\expect{\pi(P')}$ by star-edges $\estar{i}$ at level $i$ is at most $\sum_{C \in \cC} 2\card{C} \le 2n$.
    Summing over all $\log n$ levels gives an upper bound of $O(n \log n)$.

    \underline{Edges contained in level $\log n - 1$ SCCs}:
    Only edges in SCCs of size at most $2$ at level $\log n - 1$ can contribute to $\expect{\pi(P')}$ since the parts of $P'$ passing through large SCCs would have been replaced by star-edges.
    Each such edge $e$ has $\pi(e) \le 2$ and there can be at most $n$ such edges, giving a total contribution to $\expect{\pi(P')}$ of at most $2n$.

    \underline{Total}:
    Summing up the bounds gives $O(n \log^4(n) / \eps)$.
\end{proof}

Finally, we show that $P'$ approximates $P$ in length and delay.
Note that every time $P'$ uses a star-edge $uv \in \estar{i}$, it incurs as much as $D_i$ additive error (in both length and delay), because all star-edges in $\estar{i}$ have $\ell(uv) = d(uv) = D_i$, while the $(u,v)$-subpath of $P$ could be arbitrarily short.
We need to show that the sum of all these errors is no more than $O(n \log n)$.

\begin{observation}
\label{obs:dense-error}
    $\ell(P') \le (1 + O(\eps \log n))n/\eps$ and $d(P') \le (1 + O(\eps \log n))n/\eps$.
\end{observation}
\begin{proof}
    First, note that there are at most $n/2$ large SCCs at level $\log n - 1$ of the LDD Hierarchy.
    The star-edges in $\estar{\log n - 1}$ that were added to $P'$ from large SCCs thus accumulate a total error of $2n$.

    Next, note that there are at most $2n/D_i$ SCCs at level $i$ of the LDD Hierarchy that are long for $P$ since $\ell(P) \le n / \eps$ and $d(P) \le n / \eps$.
    The star-edges in $\estar{i}$ that were added to $P'$ from long SCCs thus accumulate a total error of $4n$.

    In aggregate over all $\log n$ levels, there is thus an additive $O(n \log n)$ error in length and delay.
    That is, $\ell(P') \le \ell(P) + O(n \log n) \le n / \eps + O(n \log n) = (1 + O(\eps \log n))n/\eps$.
    A similar calculation gives the same bound for $d(P')$.
\end{proof}

\subsection{Putting the Pieces Together}

Here we combine the results of this section into a proof of the following theorem that implies~\Cref{thm:dense-simple}.
\begin{restatable}[$\rsp$ on Dense Graphs]{theorem}{maindense}
\label{thm:dense}
    There is a $(1 + \eps, 1 + \eps)$-approximate algorithm for $\rsp$ that runs in $O\paren{n^2 \log^9(n) \log(nW) / \eps^3}$ time, where $W$ is the aspect ratio of edge lengths.
    The algorithm is Monte Carlo randomized and its output is correct with high probability.
\end{restatable}
\begin{proof}\ 
    \begin{itemize}
        \item Constructing $H$ and $\pi$ takes $T = O\paren{(m + n \log\log n) \log^3 n}$ since each iteration of the main loop which constructs the LDD Hierarchy is dominated by $O\paren{(m + n \log\log n) \log^2 n}$ by \Cref{prop:ldd,prop:scc}.
        \item By \Cref{obs:dense-H-pfp}, $H$ is \PFP.
        \item By \Cref{obs:dense-Pi}, $\Pi = O(n \log n)$.
        \item By \Cref{lem:dense-pi,obs:dense-error}, the \PSCs are $4$-satisfied by $\pi,H$ for all $t \in V$.
    \end{itemize}
    Therefore, \Cref{lem:gap-solver} allows us to solve Gap $\rsp$ in time $O\paren{n^2 \log^5 n / \eps^2}$.
    Plugging this solver for Gap $\rsp$ into \Cref{lem:reduction-to-gap} allows us to then get a $(1 + \eps, 1 + \eps)$-approximate solution to $\rsp$ in time $O\paren{n^2 \log^9(n) \log(nW) / \eps^3}$, completing \Cref{thm:dense}.
\end{proof}

\section{$\rsp$ on Sparse Directed Graphs}
\label{sec:sparse}

This is the most technical section of the paper, which is devoted to proving \Cref{thm:sparse-simple}.
We first sketch an algorithm for DAGs that runs slightly faster than $\Ot(mn^{2/3})$ for constant $\eps$, with the aim of introducing some new ideas that are used to prove \Cref{thm:sparse-simple}.
Then, we combine these ideas with those in \Cref{sec:dense} to generalize the algorithm for DAGs to arbitrary directed graphs.

\subsection{Motivation: $\rsp$ on Sparse Directed Acyclic Graphs}
For this subsection, let us ignore most of the machinery spelled out in \Cref{sec:framework} and focus on a $(1, 1 + \eps)$-approximate algorithm for DAGs.
This algorithm will be to the ideas of \Cref{thm:sparse-simple} as Example 2 in \Cref{subsec:examples} is to the ideas of \Cref{thm:dense-simple}.

\begin{tbox}
    \underline{\textbf{\textsc{$(1, 1 + \eps)$-approximate Algorithm for $\rsp$ on Sparse DAGs}}}
    \begin{enumerate}
        \item Topologically sort the vertices.
        \item Partition the vertices into contiguous (in the topological order) blocks $B_1, B_2, \ldots, B_{O(n/\blocksize)}$ of size $\blocksize$.
        \item \label{alg:dag-3} In each block $B_i$:
            \begin{enumerate}
                \item Randomly sample, with replacement, $\Theta\paren{\frac{\blocksize \log n}{\blockrun}}$ vertices $v \in B_i$ into $S_{B_i}$.
                \item Preprocess $(1, 1 + \eps)$-approximate All-Pairs $\rsp$ on $G[B_i]$ with delay thresholds in the range $[\delta, D]$, where $\delta = \eps D / n^2$ (recall $\delta$ defined in \Cref{sec:dp}, this is analogous).
                \item For all $(1 + \eps)^k \in [\delta, D]$ and for all $u, v \in S_{B_i}$, add the edge $uv$ with
                    \begin{align*}
                        & \ell(uv) \gets \textrm{ the length returned by All-Pairs } \rsp \textrm{ with delay } (1 + \eps)^k
                        \\
                        & d(uv) \gets (1 + \eps)^{k+1}
                    \end{align*}
                    into the multiset $H_{B_i}$.
            \end{enumerate}
        \item Set
            \begin{align*}
                \pi(uv) = 
                    \begin{cases}
                        \blocksize & \textrm{if } u \in B_i, v \in B_j \textrm{ for } B_i \neq B_j \quad \textcolor{blue}{\texttt{(inter-block)}}\\
                        \blocksize / \blockrun & \textrm{if } uv \in E(B_i) \textrm{ for some } i \quad \textcolor{blue}{\texttt{(intra-block)}}\\
                        \blocksize & \textrm{if } uv \in \bigcup_i H_{B_i}.
                    \end{cases}
            \end{align*}
        \item $\hat{G} \gets G \cup \paren{\bigcup_i H_{B_i}}$.
        \item Run $\pidp{\pi}$ on $\hat{G}$ up to depth $h = O(n)$ with delay threshold $(1 + \eps)D$.
    \end{enumerate}
\end{tbox}

For now it would be instructive to think of $\blockrun^2 = \blocksize = n^{2/3}$, which would yield an $\Ot(mn^{2/3})$ time $(1, 1 + \eps)$-approximate algorithm for $\rsp$.
Optimizing $\blocksize$ and $\blockrun$ more carefully allows us to obtain slightly faster running times across different ranges of edge densities $n^{1 + \alpha}$ for any real $\alpha \in [0, 1/2]$.

We now provide an informal sketch as to why the algorithm is correct, and how to compute its running time.
The goal here is to impart intuition; when it comes to our result on arbitrary directed graphs later, the proofs will be formal.

\begin{proof}[Sketch of Correctness]
\renewcommand{\qedsymbol}{``$\square$''}
    First, note that $\paren{\bigcup_i H_{B_i}}$ is \PFP, so the call to $\pidp{\pi}$ on $\hat{G}$ is sound.
    Let $P$ be a shortest $(s,t)$-path in $G$ with $d(P) \le D$.
    We will show the existence of an $(s,t)$-path $P'$ in $\hat{G}$ with $\ell(P') \le \ell(P)$ and $d(P') \le (1 + \eps)D$ and $\pi(P') = O(n)$.
    The call to $\pidp{\pi}$ then returns the length of an $(s,t)$-path $P''$ with $\ell(P'') \le \ell(P') \le \ell(P)$ and $d(P'') \le (1 + \eps)^2 D$.
    By rescaling $\eps$ by a constant factor, we can pretend that $d(P'') \le (1 + \eps) D$ and so the algorithm above gives $(1, 1 + \eps)$-approximate solutions to $\rsp$.
    It remains to show the existence of this path $P'$.

    \underline{Construction of $P'$}:
    Since $G$ is a DAG, notice that once $P$ leaves a block it can never reenter the same block.
    Fix an arbitrary block $B_i$.
    Let $Q$ be the section of $P$ passing through $B_i$.
    With high probability, there are vertices $u,v \in S_{B_i}$ such that $u$ is among the first $\blockrun$ vertices of $Q$ and $v$ is among the last $\blockrun$ vertices of $Q$.
    We replace the subpath of $Q$ from $u$ to $v$ with delay $d$ such that $(1 + \eps)^{k-1} \le d \le (1 + \eps)^k$ with the edge $uv \in H_{B_i}$ with delay $(1 + \eps)^k$; the resulting subpath thus uses at most $2\blockrun$ edges from $G$ and one edge from $H$.
    Repeating this over all blocks yields the $(s,t)$-path $P'$ with delay dilated by $(1 + \eps)$.
    Finally, note that:
    \begin{align*}
        \pi(P') = 
        \underbrace{O(n/\blocksize)}_{\substack{\textrm{Number of} \\ \textrm{inter-block}\\ \textrm{edges}}}
        \cdot 
        \blocksize
        + 
        \underbrace{O(\blockrun)}_{\substack{\textrm{Number of} \\ \textrm{intra-block}\\ \textrm{edges per block}}}
        \cdot
        \underbrace{O(n/\blocksize)}_{\substack{\textrm{Number of} \\ \textrm{blocks}}}
        \frac{\blocksize}{\blockrun}
        +
        \underbrace{O(n/\blocksize)}_{\substack{\textrm{Number of} \\ H\textrm{ edges}\\\textrm{per block}}}
        \cdot
        \blocksize
        = O(n).
    \end{align*}
    
\end{proof}
\begin{proof}[Sketch of Running Time]
\renewcommand{\qedsymbol}{``$\square$''}
    The running time of the algorithm is the sum of the time taken to compute All-Pairs $\rsp$ and construct $H$ in each block and the time taken by the call to $\pidp{\pi}$. We require the following more precise statement of~\Cref{thm:all-pairs-simple} proved in~\Cref{sec:all-pairs}.
\begin{restatable*}[All-Pairs $\rsp$]{theorem}{thmallpairs}\label{thm:all-pairs}
    Let $D_{\min}\leq D_{\max}$ be such that $D_{\max}/D_{\min}=O(\poly(n))$.
    After preprocessing $G$ in
    $O(mn\log^3(n)/\eps+n^2\log^5(n)/\eps^2)$ time,
    one can answer the following queries in constant time. Given $s,t\in V$ and $D\in [D_{\min},D_{\max}]$, compute the length of some $(s,t)$-path $P_{s,t}$ in $G$ such that $\ell(P_{s,t})\leq \rdist_G(s,t,D)$
    and $d(P_{s,t})\leq (1+\eps)\cdot D$.
    The data structure is Monte Carlo randomized, and the answers it produces are correct with high probability.
\end{restatable*}

    Computing All-Pairs $\rsp$ takes $\Ot(\card{E(B_i)}\blocksize)$ time on block $B_i$ by \Cref{thm:all-pairs} and so $\Ot(m \blocksize)$ in total.
    Computing $H$ takes $\Ot(\blocksize^2/\blockrun^2)$ time per block and thus $\Ot(n\blocksize/\blockrun^2)$ in total.
    Finally, with the aim of calculating the time taken by $\pidp{\pi}$, we compute $\Pi$ (see \Cref{thm:pi-dp}).
    \begin{align*}
        \Pi =
        \underbrace{O(m)}_{\substack{\textrm{Number of} \\ \textrm{inter-block}\\ \textrm{edges}}}
        \frac{1}{\blocksize}
        +
        \underbrace{O(m)}_{\substack{\textrm{Number of} \\ \textrm{intra-block}\\ \textrm{edges}}}
        \frac{\blockrun}{\blocksize}
        +
        \underbrace{\Ot(n\blocksize/\blockrun^2)}_{\card{H}}
        \frac{1}{\blocksize}.
    \end{align*}
    Plugging this in to \Cref{thm:pi-dp} up to depth $h = O(n)$, setting $\blockrun^2 = \blocksize = n^{2/3}$, and accounting for the time taken to compute All-Pairs $\rsp$ and $H$ yields a running time of $\Ot(mn^{2/3})$.
\end{proof}
The final idea which we have not gone over, but mentioned in passing, is to optimize $\blocksize$ and $\blockrun$;
on a graph with $n^{1 + \alpha}$ edges, we can set $\blocksize = n^x$ and $\blockrun = n^y$ with $y \le x$ and, ignoring $\log$ and $\eps$ factors in the calculation, solve a system over $x$ and $y$.
For example, when $\alpha = 0$, this gets us an $\Ot(mn^{3/5})$ time algorithm, which is marginally better than $\Ot(mn^{2/3})$.
Notice that $\blockrun^2 = \blocksize = n^{2/3}$ is a very reasonable choice since then $\card{H} = \Ot(n)$; we do not increase the edge density of the graph we pass into $\pidp{\pi}$.
Counter-intuitively, when $\alpha = 0$, we will have $\card{H} = n^{6/5} > n = \card{E(G)}$ so we are passing into $\pidp{\pi}$ a denser graph.
This increase in density is counterbalanced by the fact that $P'$ would use fewer edges per block, thus allowing us to set higher values for $\pi(e)$ for some of the edges inside the blocks, which would result in a decrease in $\Pi$.

Having acquired some ideas for sparse DAGs, we are finally ready to work on proving \Cref{thm:sparse-simple}.

\subsection{High Level Differences With the Dense LDD Hierarchy}

Recall now from \Cref{sec:framework} and, in particular, \Cref{lem:gap-solver}, our goal is to efficiently construct an auxiliary edge set $H$ that is \PFP and a frequency function $\pi$ so that $\pi, H$ satisfy the \PSCs for all $t \in V$.

Similarly to \Cref{sec:dense}, we will use an LDD Hierarchy to define $\pi$ and $H$.
However, the LDD Hierarchy here is different and substantially more complicated than what we have seen so far.

For starters, there will be six types of edges to be dealt with: \[\edead,\ \ehop,\ \eintra,\ \eback{i},\ \eforward{i},\ \estar{i}\] where the last three types are graded into $O(\log n)$ levels; we will define them later.
Next, there will be three types of SCCs in each level: small, medium, and large.
Finally, the LDD Hierarchy will be constructed in two phases:
\begin{enumerate}
    \item Phase 1: Not so much unlike the LDD Hierarchy in \Cref{sec:dense}, but we only recurse into medium SCCs.
    \item Phase 2: Go through each level and create blocks among the small SCCs, compute All-Pairs $\rsp$ in each block, and add auxiliary edges between sampled vertices to ``shortcut'' the block.
        This is analogous to what we have seen in the algorithm above for sparse DAGs.
\end{enumerate}

\subsection{The Sparse Directed Low-Diameter Decomposition Hierarchy}

Throughout this section, there will be two parameters whose values are to be determined later:
\begin{align*}
    \blockrun \le \blocksize.
\end{align*}
They serve similar roles as they did in the algorithm above for sparse DAGs. It will be instructive to think of $\blockrun^2 = \blocksize = n^{2/3}$ for now, which yields an algorithm with runtime $\Ot(mn^{2/3})$; the final setting will be slightly different.

Before proceeding, we make some simple definitions.

\begin{definition}[Small SCCs]
\label{def:sparse-small}
    An SCC at $C$ of the LDD Hierarchy is said to be \emph{small} if ${\card{C} \le \blocksize}$.
\end{definition}

We remind the reader of the notion of large SCCs (same as \Cref{def:dense-large}).

\begin{definition}[Large SCCs]
\label{def:sparse-large}
    An SCC $C$ at level $i$ of the LDD Hierarchy is said to be \emph{large} if $\card{C} \ge D_i = n/2^i$.
\end{definition}

\begin{definition}[Medium SCCs]
\label{def:sparse-med}
    An SCC $C$ at level $i$ of the LDD Hierarchy is said to be \emph{medium} if it is neither small nor large (at level $i$).
\end{definition}

Finally, recall the combined length-delay graph (see \Cref{def:lendel}).
And with this, we can define the first phase for constructing the LDD Hierarchy.

\clearpage
\begin{tbox}
    \algname{Low-Diameter Decomposition Hierarchy Phase 1}\\
    \textbf{Input:} A directed graph $G = (V, E, \ell, d)$, recursion depth $i$.\\
    \textbf{Output:} Data of a partial LDD Hierarchy.
    \begin{enumerate}
        \item If $i = O(\log (n / \blocksize))$, then: \textcolor{blue}{\texttt{(defining level $i$ data)}}
            \begin{enumerate}
                \item $D_i \gets n/2^i$.
                \item Add $\LDD(\lendel{G}, D_i)$ to $\eback{i}$.
                \item $\cC = \set{C_1, C_2, \ldots, C_k} \gets \SCC(G \setminus \eback{i})$.
                \item Add $(G \setminus \eback{i}) \setminus \bigcup_{C \in \cC} E(C)$ to $\eforward{i}$.
                \item For all SCCs $C$, add star-edges with length and delay $D_i$ to $\estar{i}$.
                \item For all large SCCs $C$, add $E(C)$ to $\edead$.
                \item For all medium SCCs $C$, recurse into $G[C]$ with recursion depth $i+1$.
            \end{enumerate}
    \end{enumerate}
\end{tbox}

Star-edges above are the same as in \Cref{sec:dense}: For any SCC $C$, pick a representative vertex $c \in C$ and add the edges $cv$ and $vc$ for all $v \in C$ to $H$ with $\ell(cv) = \ell(vc) = d(cv) = d(vc) = D_i$.

Once the first phase is complete, which has formed the tree structure of the LDD, we are ready to begin the second phase.
This phase mimics the block-forming and adding of shortcuts within each block that was done in the algorithm for sparse DAGs.
We first need a couple of definitions.

\begin{definition}[Block]
\label{def:block}
    A level $i$ \emph{block} is a collection of small SCCs at level $i$.
    We sometimes abuse notation and say that a block is the union of small SCCs that it contains.
\end{definition}

We will partition the small SCCs at level $i$ into blocks in the second phase, but we also want the blocks to obey the nested structure of the LDD Hierarchy (see \Cref{fig:hierarchy-sparse}).
The following definition captures this requirement (see \Cref{fig:chopped} for a pictorial example).
\begin{figure}[hb!]
    \centering
    \includegraphics[scale=0.6]{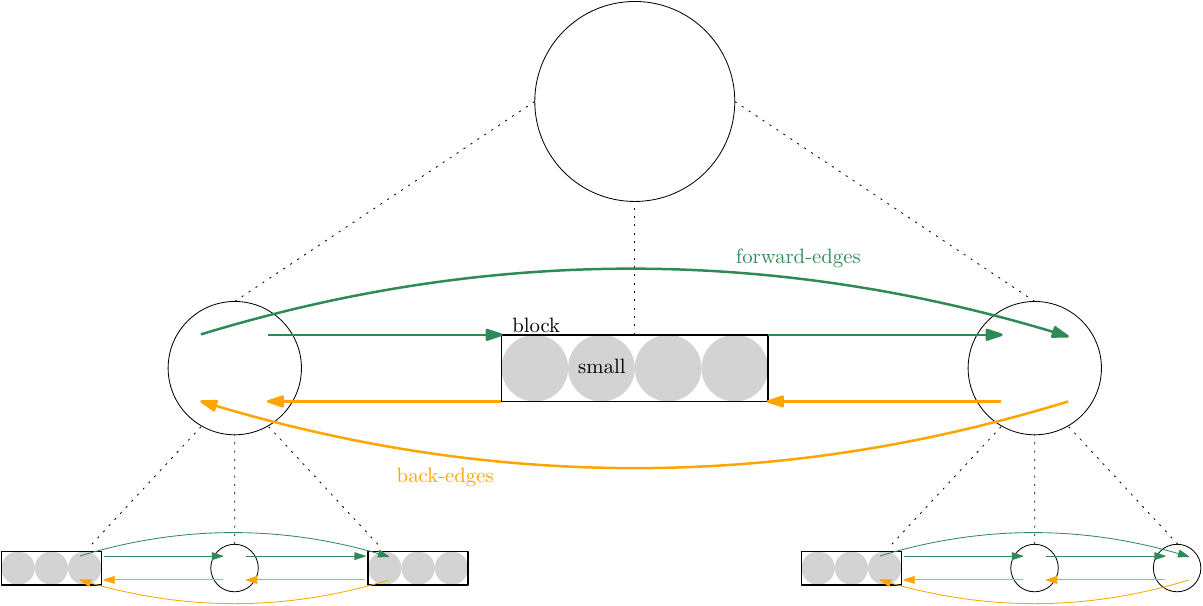}
    \caption{
        Blocks respecting the structure of the LDD Hierarchy.
        Unfilled circles represent non-small SCCs and gray blobs represent small SCCs.
        Rectangular boxes represent blocks.
        We can then talk about the ancestors of blocks or, for a fixed level, the topological ordering of blocks and non-small SCCs.
    }
    \label{fig:hierarchy-sparse}
\end{figure}

\begin{definition}[Finely Chopped Blocks]
   A level $i$ block $B$ is \emph{finely chopped} if the following holds:
    \begin{itemize}
        \item $B$ is contained in a level $i - 1$ SCC of the LDD Hierarchy.
        \item For any two SCCs $C_x$ and $C_z$ contained in $B$,  there does not exist a level $i$ SCC $C_y$ that is not small and who is topologically between $C_x$ and $C_z$. 
    \end{itemize}
\end{definition}

\begin{figure}[h]
    \centering
    \includegraphics[scale=0.7]{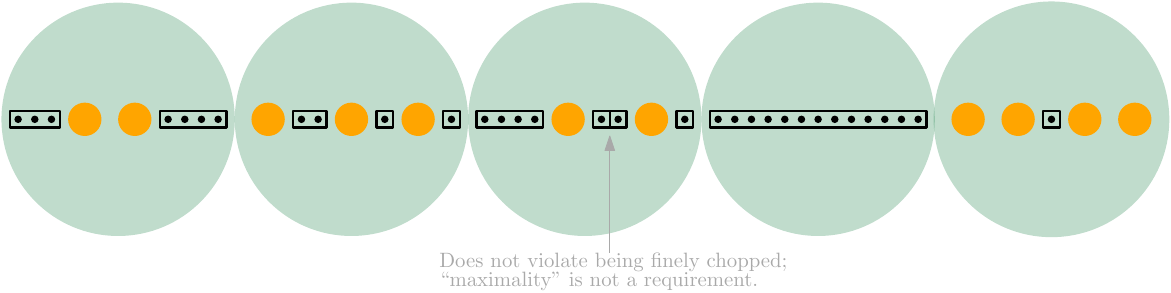}
    \caption{
        A (not drawn-to-scale) depiction of finely chopped level $i$ blocks.
        Green blobs represent level $i - 1$ SCCs, orange blobs represent level $i$ SCCs that are not small, and black dots represent small level $i$ SCCs.
        Blocks drawn with thick rectangular boxes.
    }
    \label{fig:chopped}
\end{figure}

The rough idea of the second phase is to partition the small SCCs in level $i$ into a small number of small-sized contiguous finely chopped blocks, and then do as in line \hyperref[alg:dag-3]{3} of the algorithm for sparse DAGs within each finely chopped block:
add intra-block shortcuts to discount a path's movement through the block.
This begets the question: can we efficiently partition the small SCCs in level $i$ into a small number of small-sized contiguous finely chopped blocks?
The next observation answers in the affirmative.

\begin{observation}
\label{obs:finely-chopped}
    A partition of the small SCCs of level $i$ into $O(n/\blocksize)$ contiguous finely chopped blocks of size $O(\blocksize)$ can be computed in $O(n)$ time.
\end{observation}
\begin{proof}\ 
    Consider the following process, and refer to \Cref{fig:chopping}.
    \begin{tbox}
        \begin{enumerate}[(1)]
            \item Arrange the small level $i$ SCCs in increasing topological order $C_1, C_2, \ldots, C_k$.
            \item Partition $C_1, C_2, \ldots, C_k$ into contiguous blocks $B_1, B_2, \ldots, B_{O(n/\blocksize)}$ of size $O(\blocksize)$.
            \item Use two pointers, one to walk over level $i$ blocks in increasing topological order and another to walk over level $i - 1$ SCCs in increasing topological order.
                If there is a level $i$ block $B$ and a level $i - 1 $ SCC $C$ such that $B \cap C \neq \emptyset$ and $B \not\subseteq C$, then split $B$ into two blocks $B \cap C$ and $B \cap \overline{C}$.
                \textcolor{blue}{\texttt{(there are still $O(n/\blocksize)$ blocks after this process)}}
            \item Walk over all level $i$ SCCs $C'_1, C'_2, \ldots, C'_{k'}$ in increasing topological order.
                If there is a medium or large SCC $C'_j$ such that $C'_{j-1}$ and $C'_{j+1}$ are small SCCs both belonging to the same block $B$, split $B$ into two blocks $B_{\textrm{prefix}}$ and $B_{\textrm{suffix}}$ where $B_{\textrm{prefix}}$ contains the initial SCCs of $B$ up to and including $C'_{j-1}$, and $B_{\textrm{suffix}}$ contains the remaining SCCs of~$B$.
                \textcolor{blue}{\texttt{(there are still $O(n/\blocksize)$ blocks after this process)}}
        \end{enumerate}
    \end{tbox}

    This process takes linear time.
    The obtained blocks are finely chopped by construction.
    It remains to justify our comments that after lines $3$ and $4$ there are still $O(n/\blocksize)$ blocks.

    First, note that in step (2) we can form $O(n / \blocksize)$ blocks of size $O(\blocksize)$ each since small SCCs have size no larger than $\blocksize$.
    Then, in step (3), observe that the level $i - 1$ SCCs which can contain level $i$ SCCs are medium and hence have size at least $\blocksize$.
    There are consequently $O(n / \blocksize)$ level $i - 1$ SCCs which can split the blocks from step (1), creating $O(n / \blocksize)$ new blocks.
    Finally, in step (4), notice that the level $i$ SCCs that are not small have size at least $\blocksize$.
    As before, there are $O(n / \blocksize)$ of these SCCs, creating $O(n / \blocksize)$ new blocks from splitting along them.
\end{proof}

In view of \Cref{obs:finely-chopped}, we will henceforward assume that all blocks respect the LDD structure (see \Cref{fig:hierarchy-sparse} again).
That is to say, at level $i$ there is a topological ordering over the corresponding collection of non-small SCCs and blocks, and also that we refer to an SCC containing a block as its ancestor.

The second phase is then executed as follows.

\begin{tbox}
    \algname{Low-Diameter Decomposition Hierarchy Phase 2}\\
    \textbf{Input:} The partial LDD Hierarchy from Phase 1.\\
    \textbf{Output:} Complete data of an LDD Hierarchy.
    \begin{enumerate}
        \item For each level $i$ from the Phase 1 LDD Hierarchy:
            \begin{enumerate}
                \item Partition the small SCCs of level $i$ into contiguous (in the topological order of all SCCs in level $i$) blocks $B_1, B_2, \ldots, B_{O(n/\blocksize)}$, each of which contains $O(\blocksize)$ vertices and is, moreover, finely chopped.
                \item In each block $B_j$, remove all edges $e \in E(G[B_j])$ from $\eforward{i} \cup \eback{i}$ and add them to $\eintra$.
                \textcolor{blue}{\texttt{(Relabelling edges contained in the block to $\eintra$)}}
                \item \textcolor{blue}{\texttt{Same as line \hyperref[alg:dag-3]{(3)} of the algorithm for sparse DAGs}}
                    \\In each block $B_j$:
                    \begin{enumerate}
                        \item Randomly sample, with replacement, $\Theta\paren{\frac{\blocksize \log n}{\blockrun}}$ vertices $v \in B_j$ into $S_{B_j}$.
                        \item Preprocess $(1, 1 + \eps)$-approximate All-Pairs $\rsp$ on $G[B_j]$ with delay thresholds in the range $[\delta, D]$, where $\delta = \eps D / n^2$.
                        \item For all $(1 + \eps)^k \in [\delta, D]$ and for all $u, v \in S_{B_j}$, add the edge $uv$ with
                            \begin{align*}
                                & \ell(uv) \gets \textrm{ the length returned by All-Pairs } \rsp \textrm{ with delay } (1 + \eps)^k
                                \\
                                & d(uv) \gets (1 + \eps)^{k+1}
                            \end{align*}
                            into the multiset $\ehop$.
                    \end{enumerate}
            \end{enumerate}
    \end{enumerate}
\end{tbox}

    \begin{figure}[ht!]
        \centering
        \includegraphics[scale=0.7]{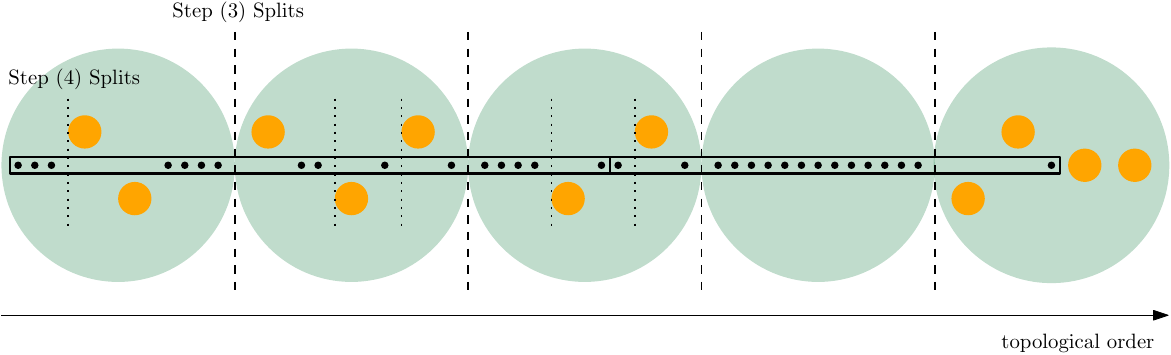}
        \caption{
            A depiction of the level $i$ block construction process.
            Green blobs represent level $i - 1$ SCCs, orange blobs represent level $i$ SCCs that are not small, and black dots represent small level $i$ SCCs.
            Two initial blocks from step (2) are drawn with thick rectangular boxes.
            Step (3) splits the initial blocks along large dashed lines.
            Step (4) splits the blocks from step (3) along smaller dotted lines.
            The resulting blocks at the end of step (4) are the same as in \Cref{fig:chopped}.
        }
        \label{fig:chopping}
    \end{figure}

\paragraph{Auxiliary Edges. }
Note that the auxiliary edge set is now $H = \ehop \cup (\bigcup_i \estar{i})$.
Note also that $H$ contains parallel edges, and we treat $H$ as disjoint from $E(G)$.
\begin{observation}
\label{obs:sparse-H-pfp}
    $H$ is \PFP.
\end{observation}
\begin{proof}
    The proof that the edges in $\bigcup_i \estar{i}$ are Pareto dominated by some path is the same as \Cref{obs:dense-H-pfp}.
    Similarly, the length of each edge $uv \in \ehop$ with $d(uv) = (1 + \eps)^k$ comes from a $(1,1+\eps)$-approximate solution to All-Pairs $\rsp$ with delay threshold $(1 + \eps)^{k-1}$, and therefore equals the length of a $(u,v)$-path $P$ with $\ell(P)\leq \dist_{G[B_j]}(u,v,(1+\eps)^{k-1})$ and $d(P) \le (1 + \eps)^k$; the edge $uv$ is consequently Pareto dominated by $P$.
\end{proof}

\subsection{The Edge Frequencies}
We set up the individual edges' frequencies as follows:
\begin{align*}
    \pi(e) =
        \begin{cases}
            \infty & \textrm{if } e \in \edead \\
            \blocksize & \textrm{if } e \in \ehop \\
            \blocksize / \blockrun & \textrm{if } e \in \eintra \\
            \blocksize & \textrm{if } e \in \eback{i} \\
            \blocksize & \textrm{if } e \in \eforward{i} \\
            \blocksize & \textrm{if } e \in \estar{i}
        \end{cases}
\end{align*}

\subsubsection{Verifying the Path-Sum Constraints}

Our strategy is similar to that of \Cref{sec:dense} on a high level.
We fix an arbitrary $(s,t)$-path $P$ in $G$ with $\ell(P) \le n/\eps$ and $d(P) \le n/\eps$ for the remainder of this section, where our aim is to exhibit an $(s,t)$-path $P'$ in $G \cup H$ which witnesses that the \PSC for $t$ is satisfied.

Recall the notion of SCCs that are long or short for P (same as \Cref{def:dense-long}).

\begin{definition}[Long and Short SCCs]
\label{def:sparse-long}
    Let $D_i = n/2^i$.
    An SCC $C$ at level $i$ of the LDD Hierarchy is said to be \emph{long} for $P$ if $\ell(P \cap G[C]) > D_i/\eps$ or $d(P \cap G[C]) > D_i/\eps$.
    Otherwise, we say that $C$ is \emph{short} for $P$.
\end{definition}

To support our later claims, we first make an addendum to this definition.
This will aid us in tallying up the contribution of some edges in level $i$ of the LDD Hierarchy which will depend on the edges in levels $j < i$.
\begin{definition}[Short All the Way Up]
\label{def:sparse-short-all-the-way}
    If $C$ is an SCC at level $i$ of the LDD Hierarchy that is short for $P$ and all its ancestor SCCs in the LDD Hierarchy are also short for $P$, then we say that~$C$ is \emph{short all the way up} for $P$.
    Similarly, if $B$ is a block at level $i$ of the LDD Hierarchy and all its ancestor SCCs in the LDD Hierarchy are short for $P$, then we say that $B$ is \emph{short all the way up} for $P$.
\end{definition}

The way we construct $P'$ is initially similar to the way we did in \Cref{sec:dense}, except we simulate a contraction (by replacing subpaths with edges in $\estar{i}$) on all large SCCs in all levels of the LDD Hierarchy (as opposed to just the last level).
After simulating contractions of large and long SCCs, we turn to the approach of constructing $P'$ provided in the sketch of the algorithm above for sparse DAGs by shortcutting portions of the path moving through a block with edges in $\ehop$.

\begin{tbox}
    \algname{Construction of $P'$ Given $P$}\\
    \textbf{Input:} An $(s,t)$-path $P$ with $\ell(P) \le n/\eps$ and $d(P) \le n/\eps$.\\
    \textbf{Output:} An $(s,t)$-path $P'$ which witnesses that the \PSC for $t$ is satisfied.
    \begin{enumerate}
        \item $P' \gets P$.
        \item For each level $i$ of the LDD Hierarchy in increasing order:
            \begin{enumerate}
                \item $\cC \gets $ SCCs of level $i$ of the LDD Hierarchy (i.e. the SCCs of $G \setminus \paren{\bigcup_{j \le i} B_j}$).
                \item For each SCC $C \in \cC$
                    that does not have an ancestor SCC that is long for $P$ or large:
                    \begin{enumerate}
                        \item Let $c$ be the representative of $C$ (i.e. the center of the star).
                        \item Let $Q$ be the subpath from $u$ to $v$, where (if they exist) $u$ is the first vertex of $P'$ in $C$ and $v$ is the last vertex of $P'$ in $C$.
                        \item If $C$ is long for $P$ or $C$ is large, then $P' \gets (P' \setminus Q) \cup \set{uc, cv}$.
                    \end{enumerate}
            \end{enumerate}
        \item For each level $i$ of the LDD Hierarchy, for each level $i$ block $B$
            that does not have an ancestor SCC that is long for $P$ or large:
            \begin{enumerate}
                \item For each maximal subpath $Q$ of $P' \cap P$ contained in $G[B]$ with at least $\blockrun$ vertices:
                    \begin{enumerate}
                        \item Let $u, v$ be vertices in $S_B$ among the first and last $\blockrun$ vertices of $Q$, respectively.
                        \item Denote the subpath of $Q$ from $u$ to $v$ with $Q_{uv}$.
                        \item Let $(1 + \eps)^{k-1} \le d(Q_{uv}) \le (1 + \eps)^k$.
                        \item Let $uv$ be the edge in $\ehop$ with $d(uv) = (1 + \eps)^k$.
                        \item $P' \gets (P' \setminus Q_{uv}) \cup \set{uv}$.
                    \end{enumerate}                    
            \end{enumerate}
        \item Output $P'$.
    \end{enumerate}
\end{tbox}

Note, like in \Cref{sec:dense}, that for any SCC $C$ in the LDD Hierarchy, $G[C] \cap P' \subseteq G[C] \cap P$.
We are now ready to verify that $P'$ witnesses the \PSC for $t$.

\paragraph{Easy Bounds. }
Let us first get some easy bounds as a warmup.
These are very similar to what we have seen in \Cref{sec:dense}.

We make use of the following two observations, where the proof of the former is essentially the same as its analogue \Cref{obs:dense-has-back-then-short}, and the proof of the latter is exactly the same as \Cref{obs:dense-ldd-edges} in \Cref{sec:dense}.

\begin{observation}
\label{obs:sparse-has-back-then-short}
    Let $C$ be an SCC at level $i - 1$ of the LDD Hierarchy, and let ${X = \eforward{i} \cup \eback{i}}$ (i.e. the forward-edges and back-edges in level $i$).
    If $G[C] \cap P' \cap X \neq \emptyset$, then $C$ is medium and short \emph{all the way up} for $P$.
\end{observation}
\begin{proof}
    We prove the contrapositive of the statement.
    If $C$ were not medium, then we would not have recursed into $G[C]$ in the construction of the LDD Hierarchy and so $G[C] \cap P' \cap X = \emptyset$.
    Suppose then that $C$ or one of its ancestors were long for $P$.
    Then $G[C] \cap P'$ would be empty or comprised of star-edges, by the construction of $P'$; that is, $G[C] \cap P' \subseteq H$.
    Since~$X$ is disjoint from $H$, we must have $G[C] \cap P' \cap X = \emptyset$.
\end{proof}
\begin{observation}
\label{obs:sparse-ldd-edges}
    Let $C$ be an SCC in level $i - 1$ of the LDD Hierarchy that is short for $P$, or $C = V$.
    \[\expect{\card{G[C] \cap \eback{i} \cap P'}} = O(\log^3(n) / \eps).\]
\end{observation}

Then, using similar arguments we have seen, we can bound the contribution of $\edead, \eback{i}, \estar{i}$ to $\expect{\pi(P')}$.

\begin{observation}
\label{obs:sparse-easy-edges}
    The contribution of $\edead, \eback{i}, \estar{i}$ over all $i$ to $\expect{\pi(P')}$ is $O(n \log^4(n) / \eps)$.
\end{observation}
\begin{proof}\
    \underline{$\edead$}:
    Since all the parts of $P$ moving through large SCCs are replaced by star-edges and $\edead$ are edges contained in large SCCs, $P'$ does not use any edge in $\edead$.
    The contribution to $\expect{\pi(P')}$ is therefore $0$.
    
    \underline{$\eback{i}$}:
    Let $\cC$ be the SCCs of level $i - 1$ of the LDD Hierarchy.
    The total contribution to $\expect{\pi(P')}$ by edges in $\eback{i}$ is at most
    \begin{align*}
        \sum_{C \in \cC} \blocksize \cdot \expect{\card{G[C] \cap \eback{i} \cap P'}}
        & \le
        \sum_{\substack{C \in \cC:\\ C \textrm{ is medium and} \\ \textrm{short for } P}} \blocksize \cdot \expect{\card{G[C] \cap \eback{i} \cap P}}
        \tag{$G[C] \cap \eback{i} \cap P' \subseteq G[C] \cap \eback{i} \cap P$ and \Cref{obs:sparse-has-back-then-short}}
        \\
        & =
        \sum_{\substack{C \in \cC:\\ C \textrm{ is medium and} \\ \textrm{short for } P}} \blocksize \cdot O(\log^3(n) / \eps)
        \tag{\Cref{obs:sparse-ldd-edges}}
        \\
        & =
        O(n \log^3(n) / \eps).
        \tag{At most $O(n/\blocksize)$ medium SCCs at level $i$}
    \end{align*}
    
    \underline{$\estar{i}$}:
    $P'$ uses $O(n/\blocksize)$ edges from $\estar{i}$ since there are $O(n/D_i)$ large SCCs in level $i$ and $O(n/D_i)$ SCCs that are long for $P$ in level $i$ and $D_i \ge \blocksize$.
    The contribution to $\expect{\pi(P')}$ is therefore $O(n)$.

    \underline{Total}:
    Summing up over the $O(\log n)$ levels of the LDD Hierarchy gives a total contribution of $O(n \log^4 n / \eps)$.
\end{proof}

\paragraph{The Forward Edges. }
We are no longer able to conveniently bound the contribution of forward-edges to $\expect{\pi(P')}$ using the fact that $\pi$ was an order difference function; the $\pi$ defined in this section is entirely different from that defined in \Cref{sec:dense}.
Instead, note that for a fixed level $i$ of the LDD Hierarchy, an edge in $\eforward{i}$ moves from a level $i$ block or non-small SCC to another level $i$ block or non-small SCC that has a higher topological order in the level $i$ ordering.
Since there are $O(n / \blocksize)$ blocks and non-small SCCs in level $i$, as we walk along $P'$ there can be at most $O(n / \blocksize)$ edges from $\eforward{i}$ before the topological order along the walk must take a step backwards, by some edge from $\eback{j}$ for $j \le i$.
The next claim quantifies how often these backsteps can happen, which will then allow us to bound the contribution of forward-edges to $\expect{\pi(P')}$.

\begin{lemma}
\label{lem:sparse-restarts}
    Let $V = C_0 \supseteq C_1 \supseteq \ldots \supseteq C_i = C$ be random SCCs where $C_j$ for $j \in [i]$ is in level~$j$ of the LDD Hierarchy, conditioned on $C$ being short all the way up for $P$.
    Let
    \[R = \bigcup_{j \in [i + 1]} \paren{\eback{j} \cap G[C_{j-1}] \cap P}.\]
    Then, $\expect{\card{R}} = O(\log^4(n) / \eps)$.\footnote{Suggestively, $R$ stands for ``Restarts''.}
\end{lemma}
\begin{proof}
    Since $C_j$ is short for $P$ for any $j \in [i]$, we have $\expect{\card{\eback{j+1} \cap G[C_{j}] \cap P}} = O(\log^3(n) / \eps)$ by \Cref{obs:sparse-ldd-edges}.
    Similarly, $\expect{\card{\eback{1} \cap G[C_0] \cap P}} = O(\log^3(n) / \eps)$ since the combined length-delay of $P$ is $O(n / \eps)$.
    Thus, using $i < \log n$,
    \begin{align*}
        \expect{\card{R}} \le \sum_{j \in [i + 1]} \expect{\card{\eback{j} \cap G[C_{j-1}] \cap P}} =  O(\log^4(n) / \eps). \qedhere
    \end{align*}
\end{proof}

With this, we are ready to bound the contribution of forward-edges to $\expect{\pi(P')}$.

\begin{lemma}
\label{lem:sparse-forward-edges}
    $\eforward{i}$ contributes at most $O(n \log^6(n) / \eps)$ to $\expect{\pi(P')}$ when summed up over all levels $i$.
\end{lemma}
Intuitively, this should follow by a simple application of linearity of an expectation summed over SCCs of the LDD Hierarchy.
However, since the SCCs (notably, the number of them) are random, some care must be taken to formalize this idea correctly.
\begin{proof}

    Let $C_1, C_2, \ldots, C_n$ be random subsets of $V$ such that
    \begin{itemize}
        \item For (a random) $w \in [n]$, $C_1, \ldots, C_w$ are all the SCCs at level $i - 1$ of the LDD Hierarchy that are medium and short all the way up for $P$, and $C_j = \emptyset$ for $j > w$.
        \item $p_1 \ge p_2 \ge \ldots \ge p_n$ are such that $C_j$ contains $p_j$ level $i$ non-small SCCs and blocks for all $j \in [n]$.
    \end{itemize}
    In other words, $C_j$ is a random variable denoting either the SCC at level $i-1$ of the LDD Hierarchy that is medium and short all the way up and contains the $j$-th largest number of level-$i$ non-small SCCs and blocks (with ties broken consistently) or $\emptyset$ if less than~$j$ such SCCs exist at level $i-1$.
    
    For all $j > w$, let $R_j = \emptyset$. 
    For all $j \in [w]$, let $R_j$ be the (random) set $R$ produced by plugging $C_j$ into $C$ of \Cref{lem:sparse-restarts};
    that is, $R_j$ is the set of back-edges in $P$ that are contained in an SCC in the LDD Hierarchy that contains $C_j$.

    Then, notice that for any $j$:
    \begin{align*}
        \sum_{\substack{e \in G[C_j] \cap \eforward{i} \cap P}} \pi(e)
        \le
        \blocksize
        \cdot
        p_j
        \cdot
        \card{R_j},
    \end{align*}
    since $\pi(e) = \blocksize$, the maximum number of forward-edges in $G[C_j] \cap \eforward{i} \cap P$ per back-edge in $P$ contained in an ancestor of $C_j$ is $p_j$ (since forward-edges connect blocks and non-small SCCs), and the number of back-edges in $P$ contained in an ancestor of $C_j$ is no more than $\card{R_j}$.

    Finally, by \Cref{obs:sparse-has-back-then-short}, the contribution of $\eforward{i}$ to $\expect{\pi(P')}$ is at most
    \begin{align*}
        \expect{\sum_{j \in [n]} \sum_{\substack{e \in G[C_j] \cap \eforward{i} \cap P'}} \pi(e)}
        & \le
        \sum_{j \in [n]} \expect{\sum_{\substack{e \in G[C_j] \cap \eforward{i} \cap P}} \pi(e)}
        \tag{$G[C_j] \cap \eforward{i} \cap P' \subseteq G[C_j] \cap \eforward{i} \cap P$}
        \\
        & \le
        \sum_{j \in [n]} \expect{\blocksize \cdot p_j \cdot \card{R_j}}
        \tag{Paragraph above}
        \\
        & \le
        \sum_{j \in [n]} \expect{\blocksize \cdot \frac{n}{j \blocksize} \cdot \card{R_j}}
        \tag{$p_j$ is decreasing in $j$ and $\sum_j p_j = O(n / \blocksize)$, the number of blocks and non-small SCCs in level $i$}
        \\
        & =
        \sum_{j \in [n]} \expect{\card{R_j}} \cdot \frac{n}{j}
        \\
        & =
        \sum_{j \in [n]} O(\log^4(n) / \eps) \cdot \frac{n}{j}
        \tag{\Cref{lem:sparse-restarts}}
        \\
        & =
        O(n \log^5 n / \eps).
    \end{align*}
    Summing up over all $O(\log n)$ levels gives us the final bound of $O(n \log^6(n) / \eps)$.
\end{proof}

\paragraph{The Blocks. }
The final step is to bound the contribution of the edges of types $\ehop, \eintra$ (i.e., the edges inside blocks) to $\expect{\pi(P')}$.
We first make an observation similar to \Cref{obs:sparse-has-back-then-short}.

\begin{observation}
\label{obs:sparse-block-then-short}
    Let $B$ be a level $i$ block of the LDD Hierarchy.
    If $(G \cup \ehop)[B] \cap P' \neq \emptyset$, then~$B$ is short all the way up for $P$.
\end{observation}
\begin{proof}
    We prove the contrapositive of the statement.
    Suppose that one of the ancestors of $B$ were long for $P$.
    Then by construction all edges from $(G \cup \ehop)[B] \cap P$ would be removed and replaced by some star-edges.
\end{proof}

We will also rely on the following standard fact~\cite{UllmanY91}.
\begin{fact}
\label{f:hitting-set}
    Let $k\in [1,n]$ be an integer.
    Let $P$ be a simple path in $G$ such that $\card{V(P)} \geq k$. 
    Suppose a subset $S \subseteq V$ is obtained by sampling $\Theta((n/k)\log{n})$ vertices from $V$ with replacement. 
    Then,~$P$ contains a vertex of $S$ with high probability.
    We can make this with high probability guarantee $(1 - n^{-c})$ for an arbitrary constant $c$ dependent on the constant hidden in the $\Theta$ notation.
\end{fact}

\begin{lemma}
\label{lem:sparse-block-edges}
    The contribution of $\ehop, \eintra$ to $\expect{\pi(P')}$ is $O(n \log^5(n) / \eps)$.
\end{lemma}
\begin{proof}
    By \Cref{f:hitting-set}, the probability that for all blocks $B$ all maximal subpaths of $P'$ in $(G \cup \ehop)[B]$ each use at most one edge from $\ehop$ and $O(\blockrun)$ edges from $\eintra$ is at least $(1 - n^{-2})$, which follows from connecting the $\Theta(\blocksize \log (n) / \blockrun)$ randomly sampled vertices in $B$ with edges in $\ehop$.
    Call this event $X$, and let us condition under it holding for now.
    For any block $B$, any maximal subpath of $P'$ moving through $B$ would contribute
    \begin{align*}
        \underbrace{\blocksize}_{\substack{\pi(e \in \ehop)}}
        +
        O(\blockrun)
        \cdot
        \underbrace{\blocksize / \blockrun}_{\substack{\pi(e \in \eintra)}}
        =
        O(\blocksize)
    \end{align*}
    to $\pi(P')$.
    The number of such subpaths through $B$ is no more than one plus the number of back-edges in ancestors of $B$ in the LDD Hierarchy.
    To see this, note that when walking along $P'$, an edge leaving~$B$ is either a back-edge contained in an ancestor SCC, or a forward-edge which makes an increasing step in the topological order of blocks and non-small SCCs at $B$'s level, and thus, in this case as well, the next return of $P'$ to $B$ can be charged to next back-edge taken that is contained in an ancestor of $B$.
    Let $R_B$ be the set of back-edges contained in ancestors of $B$.
    By \Cref{obs:sparse-block-then-short} chained with \Cref{lem:sparse-restarts}, $\expect{\card{R_B}+1} = O(\log^4 (n) / \eps)$.

    Putting all this together, the expected contribution $Q$ of $\ehop, \eintra$ to $\pi(P')$ is
    \begin{align*}
        \expect{Q}
        & =
        (1 - 1/n^2) \expect{Q\Bigg\lvert X}
        +
        (1/n^2) \expect{Q \Bigg\lvert \bar{X}}
        \\
        & \le
        \expect{Q \Bigg\lvert X} + 1
        \tag{$\pi(P') = n^2$ in the worst case}
        \\
        & =
        O\paren{\expect{\sum_B \blocksize \card{R_B} \Bigg\lvert X}}
        \\
        & =
        O\paren{\sum_B \blocksize \log^4(n) / \eps}
        \tag{Paragraph above}
        \\
        & =
        O\paren{\frac{n \log n}{\blocksize} \blocksize \log^4(n) / \eps}
        \tag{$O(n / \blocksize)$ blocks per level and $O(\log n)$ levels}
        \\
        & =
        O\paren{n \log^5(n) / \eps}.
    \end{align*}
\end{proof}

\paragraph{Bounding $\expect{\pi(P')}$. }
Tallying up the contributions yields the final claim on $\expect{\pi(P')}$.
\begin{lemma}
\label{lem:sparse-pi}
    $\expect{\pi(P')} = O(n \log^6(n) / \eps)$
\end{lemma}
\begin{proof}
    This is immediate from summing up \Cref{obs:sparse-easy-edges,lem:sparse-forward-edges,lem:sparse-block-edges}.
\end{proof}

\paragraph{The Error Accumulated. }
Finally, we show that $P'$ approximates $P$ in length and delay.
Recall that every time $P'$ uses an edge in $\estar{i}$, it incurs as much as $D_i$ additive error (in both length and delay) because all such edges have $\ell(uv) = d(uv) = D_i$, while the $(u,v)$-subpath of $P$ could be arbitrarily short.
For the LDD Hierarchy on sparse directed graphs, there is one more source of error: using the edges of $\ehop$.
Every time $P'$ uses an edge from $\ehop$, the delay contribution of the replaced subpath of $P$ is dilated by a $(1 + \eps)$ factor.
We need to show that the sum of all these errors is no more than $O(n \log n)$.

\begin{observation}
\label{obs:sparse-error}
    $\ell(P') \le (1 + O(\eps \log n))n/\eps$ and $d(P') \le (1 + O(\eps \log n))n/\eps$.
\end{observation}
\begin{proof}\ 

    \underline{Error from $\estar{i}$}:
    First, note that there are at most $n/D_i$ large SCCs at level $i$ of the LDD Hierarchy.
    The edges from $\estar{i}$ that were added to $P'$ from large SCCs thus accumulate a total error of $2n$.
    Next, note that there are at most $2n/D_i$ SCCs at level $i$ of the LDD Hierarchy that are long for $P$ since $\ell(P) \le n / \eps$ and $d(P) \le n / \eps$.
    The edges $\estar{i}$ that were added to $P'$ from long SCCs thus accumulate a total error of $4n$.

    \underline{Error from $\ehop$}:
    Shortcutting blocks by using edges in $\ehop$ weakly decreases the error in length, and dilates the delay by at most $(1 + \eps)d(Q_{uv})$, where $Q_{uv}$ was the subpath of $P \cap P'$ replaced by the edge $uv$ from $\ehop$.
    Since $d(P) \le n/\eps$, this adds at most $n$ error over all the blocks.

    \underline{Total Error}:
    In aggregate over all $\log n$ levels, there is thus an additive $O(n \log n)$ error in length and delay.
    That is, $\ell(P') \le \ell(P) + O(n \log n) \le n / \eps + O(n \log n) = (1 + O(\eps \log n))n/\eps$.
    A similar calculation gives the same bound for $d(P')$.
\end{proof}

\subsubsection{Bounding $\Pi$ and the Construction Time}

The missing ingredients needed to invoke \Cref{lem:gap-solver} are the following, which we now address.
\begin{itemize}
    \item An upper bound for $\Pi$ (recall its meaning from \Cref{thm:pi-dp}).
    \item The construction time of $\pi, H$ or, put otherwise, the LDD Hierarchy.
\end{itemize}

\begin{observation}
\label{obs:sparse-Pi}
If $\blockrun=\Omega(\log{n})$, then:
    $\Pi = O\paren{\frac{n \log^4 n}{\blockrun^2 \eps} + \frac{m \blockrun}{\blocksize}}$.
\end{observation}
\begin{proof}
    We bound the contribution of each type of edge, and then sum them up.

    \underline{$\edead$}:
    These contribute nothing to $\Pi$.

    \underline{$\ehop$}:
    In each block, there are $\frac{\blocksize \log n}{\blockrun}$ vertices incident to edges in $\ehop$.
    Each pair of such vertices has $O(\log(n) / \eps)$ parallel edges from $\ehop$ between them since we have added one edge per delay $(1 + \eps)^k$ in the range $[\delta,D]$.
    Since there are $O(n / \blocksize)$ level $i$ blocks and $O(\log n)$ levels, the contribution of $\ehop$ to $\Pi$ is
    \begin{align*}
        \underbrace{O(\log(n) / \eps)}_{\substack{\textrm{Edges per pair}}}
        \cdot
        \underbrace{\paren{\blocksize \log(n) /\blockrun}^2}_{\substack{\textrm{Pairs per block}}}
        \cdot
        \underbrace{O(n / \blocksize)}_{\substack{\textrm{Blocks per level}}}
        \cdot
        \underbrace{O(\log n)}_{\substack{\textrm{Levels}}}
        \cdot
        \underbrace{1/\blocksize}_{1/\pi(e)}
        =
        O\paren{\frac{n \log^4 n}{\blockrun^2 \eps}}.
    \end{align*}

    \underline{$\eintra$}:
    The total contribution of $\eintra$ to $\Pi$ is
    \begin{align*}
        \underbrace{O(m)}_{\textrm{Edges}}
        \cdot
        \underbrace{\blockrun/\blocksize}_{1/\pi(e)}
        =
        O\paren{\frac{m \blockrun}{\blocksize}}.
    \end{align*}

    \underline{$\eback{i}$}:
    The total contribution of $\eback{i}$ to $\Pi$ is
    \begin{align*}
        \underbrace{O(m)}_{\textrm{Edges}}
        \cdot
        \underbrace{1/\blocksize}_{1/\pi(e)}
        =
        O\paren{\frac{m}{\blocksize}}.
    \end{align*}
    
    \underline{$\eforward{i}$}:
    The total contribution of $\eforward{i}$ to $\Pi$ is
    \begin{align*}
        \underbrace{O(m)}_{\textrm{Edges}}
        \cdot
        \underbrace{1/\blocksize}_{1/\pi(e)}
        =
        O\paren{\frac{m}{\blocksize}}.
    \end{align*}
    
    \underline{$\estar{i}$}:
    The total contribution of $\estar{i}$ to $\Pi$ is
    \begin{align*}
        \underbrace{O(n)}_{\textrm{Edges}}
        \cdot
        \underbrace{1/\blocksize}_{1/\pi(e)}
        =
        O\paren{\frac{n}{\blocksize}}.
    \end{align*}

    Note that the last three contributions (over all levels) are dominated by the contribution of $\eintra$ and, therefore,
    $\Pi = O\paren{\frac{n \log^4 n}{\blockrun^2 \eps} + \frac{m \blockrun }{\blocksize}}$.
\end{proof}

\begin{observation}
\label{obs:sparse-construction-time}
    Constructing $H$ and $\pi$ takes $O(m \blocksize \log^6 n / \eps^2)$ time.
\end{observation}
\begin{proof}
    The first phase of constructing the LDD Hierarchy takes $O\paren{(m + n \log\log n) \log^3 n}$ since each iteration of the main loop which constructs the LDD Hierarchy is dominated by the \linebreak
    $O\paren{(m + n \log\log n) \log^2 n}$ cost by \Cref{prop:ldd,prop:scc}.
    This turns out to be dominated by the second phase, which comprises of three computationally demanding steps over $O(\log n)$ phases:
    \begin{itemize}
        \item Constructing finely chopped blocks.
            This takes $O(n)$ time by \Cref{obs:finely-chopped}.
        \item Running the algorithm for All-Pairs $\rsp$ in each block.
            The computation on each block $B_j$ of size $O(\blocksize)$ takes $O(\card{E_G(B_j)} \blocksize \log^5(n) / \eps^2)$ time by \Cref{thm:all-pairs}.
            Since the blocks are disjoint, this totals to $O(m \blocksize \log^5 n / \eps^2)$ running time.
        \item Adding edges to $H$.
            This takes $O(n \blocksize \log^4 n / (\blockrun^2 \eps))$ time (see $\eintra$ calculation in \Cref{obs:sparse-Pi}).
    \end{itemize}
    The dominant term is the one for computing All-Pairs $\rsp$, and so the time to construct~$H$ and $\pi$ is $O(m \blocksize \log^6(n) / \eps^2)$.
\end{proof}

\subsection{Putting the Pieces Together}

We are finally poised to prove the following theorem, which implies~\Cref{thm:sparse-simple}.
\begin{restatable}[$\rsp$ on Sparse Graphs]{theorem}{mainsparse}
\label{thm:sparse}
    Let $m = n^{1 + \alpha}$, for $\alpha \in [0,1/2]$.
    There is a ${(1+\eps,1+\eps)}$\nobreakdash-approximate algorithm for $\rsp$ that runs in $O\paren{mn^{(3 - \alpha)/5} \log^{15}(n) \log(nW) / \eps^4}$ time, where $W$ is the aspect ratio of lengths.
    The algorithm is Monte Carlo randomized and its output is correct with high probability.
\end{restatable}
\begin{proof}
    For a graph with $m = n^{1 + \alpha}$, where $\alpha \in [0,1/2]$, set
    \[\blockrun = n^{(1 - 2\alpha)/5} \textrm{ and } \blocksize = n^{(3 - \alpha)/5}.\]
    Note that $\blockrun \le \blocksize$, which is what we require for the second phase of constructing the LDD Hierarchy to make sense.
    \begin{itemize}
        \item Constructing $H$ and $\pi$ takes $O(m \blocksize \log^6 n / \eps^2)$ time by \Cref{obs:sparse-construction-time}.
            This resolves to $O(n^{(8 + 4\alpha)/5} \log^6 n / \eps^2)$ time.
        \item By \Cref{obs:sparse-H-pfp}, $H$ is \PFP.
        \item By \Cref{obs:sparse-Pi}, $\Pi = O\paren{\frac{n \log^4 n}{\blockrun^2 \eps} + \frac{m \blockrun}{\blocksize}}$.
            This resolves to $O\paren{n^{(3 + 4\alpha)/5} \log^4(n) / \eps}$ time.
        \item By \Cref{lem:sparse-pi,obs:sparse-error}, the \PSCs are $6$-satisfied by $\pi,H$ for all $t \in V$.
    \end{itemize}
    Plugging the above into \Cref{lem:gap-solver} allows us to solve Gap $\rsp$ in time \linebreak
    $O\paren{n^{(8 + 4\alpha)/5} \log^{10}(n) / \eps^3}$.
    Plugging this solver for Gap $\rsp$ into \Cref{lem:reduction-to-gap} allows us to then get a $(1 + \eps, 1 + \eps)$-approximate solution to $\rsp$ in $O\paren{n^{(8 + 4\alpha)/5} \log^{15}(n) \log(nW) / \eps^4}$ time,
    which resolves to $O\paren{mn^{(3 - \alpha)/5} \log^{15}(n) \log(nW) / \eps^4}$,
    completing \Cref{thm:sparse}.
\end{proof}

\section{All-Pairs $\rsp$}
\label{sec:all-pairs}
In this section, we give a $(1,1+\eps)$-approximate solution to the All-Pairs $\rsp$ problem running in near-optimal $\Ot(mn/\eps+n^2/\eps)$ time.

Specifically, given $D_{\min},D_{\max}\in \mathbb{R}_{\geq 0}$ such that $D_{\max}/D_{\min}=\poly(n)$, we show how to preprocess the graph so that for any $s,t\in V$ and $D\in [D_{\min},D_{\max}]$, in $O(1)$ time we can compute the length of some path $P_{s,t}$ such that $\ell(P_{s,t})\leq \rdist(s,t,D)$ and $d(P_{s,t})\leq (1+\eps)D$.
We first show how to achieve a weaker delay error guarantee of $d(P_{s,t})\leq (1+\eps)^{O(\log{n})}\cdot D$.
By decreasing the error parameter $\eps$ by a factor $\Theta(\log{n})$ we will get the desired accuracy with only polylogarithmic overhead in the running time.

Additionally, let us assume without loss of generality that all delays are at least $\eps D_{\min}/n$.
Otherwise, one could increase all the smaller delays to that value and this would incur at most $\eps D_{\min}$ additive delay error on any simple path in $G$.
Again, since our goal is obtaining merely $(1+\eps)^{O(\log{n})}$ multiplicative error on the delay, this is acceptable. 

\paragraph{Dynamic programming.}
The algorithm can be again viewed as an application of dynamic programming.
It combines dynamic programming with the technique of shortcuts, where the shortcuts we add are analogous to those used in~\cite{Bernstein16} for the different problem of Decremental All-Pairs Shortest Paths. 

Let $q=\lceil \log_2{n}\rceil$.
Let $V_0:=V$.
For $i=1,\ldots,q-1$, let $V_i$ be obtained by sampling $\Theta((n/2^i)\log{n})$ vertices of $V$ with replacement.

Let us also put $\alpha_{\min}:=\lfloor\log_{1+\eps}(\eps D_{\min}/n)\rfloor$ and $\alpha_{\max}:=\lceil \log_{1+\eps}(D_{\max})\rceil$.
Note that \linebreak ${\alpha_{\max}-\alpha_{\min}}=O(\log(n)/\eps)$.

\newcommand{\ap}{R}

For each $k=q-1,\ldots,0$, we will subsequently compute the values $\ap_k(s,t,i)$, where\linebreak $s,t\in (V_k\times V)\cup (V\times V_k)$ and $i\in\{\alpha_{\min},\ldots,\alpha_{\max}\}$, defined as follows:
\begin{align*}
    \ap_k(s,t,i):=& \text{ the length of some }(s,t)\text{-path }P_{s,t}\text{ such that }\\
    & \ell(P_{s,t})\leq\rdist(s,t,(1+\eps)^i)\text{ and }d(P_{s,t})\leq (1+\eps)^{i+2(q-k)}.
\end{align*}
Note that since $V_0=V$, the values $R_0(s,t,i)$ are sufficient to achieve our goal.
Indeed, given some $s,t\in V$, and $D\in [D_{\min},D_{\max}]$, it is enough to output $R_0(s,t,j)$ for $j=\lceil \log_{1+\eps}(D)\rceil$.
Then the corresponding path $P_{s,t}$ satisfies
\[\ell(P_{s,t})\leq \rdist(s,t,(1+\eps)^j)\leq \rdist(s,t,D)\]
and
\[d(P_{s,t})\leq (1+\eps)^{j+2q}\leq (1+\eps)^{2q+1}D=(1+\eps)^{O(\log{n})} D,\]
as desired.

We will only show how to compute the values $\ap_k(s,t,i)$ for pairs $(s,t)\in V_k\times V$.
For every subsequent $k$, the DP values for $(s,t)\in V\times V_k$ are handled symmetrically, by possibly running a part of the computations on the \emph{reverse graph} $G^R$, i.e., $G$ with all its edge directions reversed.

The algorithm will run $\pidp{\one}$ from~\Cref{sec:dp} multiple times on the graph $G$ with some appropriate auxiliary edges (shortcuts) added, for various sources $s$ and values $h$.
In all these runs we will set $D_{\min}=\eps\cdot (1+\eps)^{\alpha_{\min}}$ and $D_{\max}=(1+\eps)^{\alpha_{\max}+O(\log{n})}$ (see~\Cref{sec:dp}); note that the ratio of these boundary values is polynomial in $n$, as required by~\Cref{thm:pi-dp}.

For $k=q-1$, we simply run $\pidp{\one}$ for all sources $s\in V_{q-1}$ and $h=n$.
By~\Cref{thm:pi-dp}, this costs $O\left(|V_{q-1}|mn\log(n)/\eps\right)=O(mn\log^2(n)/\eps)$ time.
Moreover, since $\dist^n(s,t,x)=\dist(s,t,x)$ for any $x$, observe that it allows obtaining precisely the sought values $\ap_{q-1}(s,t,i)$.

\paragraph{Transitions.}
Let $k<q-1$. We now show how the values $\ap_k(s,t,i)$ are computed based on the values $\ap_{k+1}(\cdot,\cdot,\cdot)$. 

For any $s\in V_k$, let $G_s$ be the graph obtained from $G$ by adding \emph{shortcuts} from $s$ to $V_{k+1}$, as follows.
For each $v\in V_{k+1}$, and $i=\alpha_{\min},\ldots,\alpha_{\max}$, we add an edge $e_{v,i}=sv$ with
\[\ell(e_{v,i})=\ap_{k+1}(s,v,i)\text{ and }d(e_{v,i})=(1+\eps)^{i+2(q-(k+1))}.\]
In other words, $e_{v,i}$ shortcuts the $(s,v)$-path corresponding to the entry $\ap_{k+1}(s,v,i)$.
Note that $(s,v)\in V\times V_{k+1}$, so indeed $\ap_{k+1}(s,v,i)$ is well-defined and available by the inductive assumption.

With the graph $G_s$ in hand, we invoke the $\pidp{\one}$ with source $s$ and $h=2^{k+1}$ on the graph~$G_s$.
By~\Cref{thm:pi-dp}, using that, for each $v\in V$, in constant time we can set $\ap_k(s,v,i)$ to be the length of some $(s,v)$-path $Q_{s,v}$ in $G_s$ such that
\[\ell(Q_{s,v})\leq \rdist^{2^{k+1}}_{G_s}\left(s,v,(1+\eps)^{i+2(q-k)-1}\right)\text{ and }
d(Q_{s,v})\leq (1+\eps)^{i+2(q-k)}.\]
Note that such a path $Q_{s,v}$ in $G_s$ corresponds to some underlying path $P'_{s,v}$ in $G$: by the inductive assumption and the definition of $\ap_{k+1}(s,\cdot,\cdot)$, the shortcut edge (if used) may be expanded to a respective path it shortcuts.
The following lemma establishes the correctness of this approach.

\begin{lemma}
    For any $s\in V_k$, $v\in V$, and $i\in [\alpha_{\min},\alpha_{\max}]$, $\ell(P_{s,v}')\leq \rdist(s,v,(1+\eps)^i)$ holds with high probability.
\end{lemma}
\begin{proof}
    It is enough to prove
    \[\rdist^{2^{k+1}}_{G_s}\left(s,v,(1+\eps)^{i+2(q-k)-1}\right)\leq \dist(s,v,(1+\eps)^i).\]
    Consider a simple $(s,v)$-path $P^*\subseteq G$ such that $d(P^*)\leq (1+\eps)^i$ and $\ell(P^*)=\rdist(s,v,(1+\eps)^i)$.

    First, suppose $P^*$ has no more than $2^{k+1}$ edges.
    Then, by $P^*\subseteq G_s$ we have:
    \[ \rdist^{2^{k+1}}_{G_s}\left(s,v,(1+\eps)^{i+2(q-k)-1}\right)\leq \rdist^{2^{k+1}}_G(s,v,(1+\eps)^i)=\rdist(s,v,(1+\eps)^i).\]

    So let us assume that $P^*$ has more than $2^{k+1}$ edges.
    Consider the suffix $S$ of $P^*$ consisting of $2^{k+1}-1$ edges, so that $|V(S)|=2^{k+1}$.
    By~\Cref{f:hitting-set}, with high probability there exists a vertex $x\in V(S)\cap V_{k+1}$.
    Express $P^*=P_1P_2$, where $P_1$ is a non-empty $s\to x$ path. Note that $P_2$ has at most $2^{k+1}-1$ edges.
    Let $z:=\lceil \log_{1+\eps}d(P_1)\rceil \leq \alpha_{\max}$.
    Since $d(P_1)\geq \eps D_{\min}/n$, we also have $z\geq \alpha_{\min}$.
    Since $x\in V_{k+1}$, there is a shortcut edge $e_{x,z}$ in $G_s$, and we have
    \[\ell(e_{x,z})=\ap_{k+1}(s,x,z)\text{ and }d(e_{x,z})=(1+\eps)^{z+2(q-(k+1))}.\]
    By the definition of $\ap_{k+1}(s,x,z)$ and the optimality of $P_1$, $e_{x,z}$ certifies the existence of an $(s,x)$\nobreakdash-path $P'_1$ in $G$ with $d(P_1')\leq d(e_{x,z})$ such that
    \[\ell(e_{x,z})=\ell(P'_1)\leq \rdist(s,x,(1+\eps)^{z})\leq \rdist(s,x,d(P_1))=\ell(P_1).\]
    Consequently, we obtain:
    \begin{equation}\label{eq:apsp1}
        \ell(e_{x,z}\cdot P_2)=\ell(P_1'P_2)\leq \ell(P_1P_2)=\ell(P^*)=\rdist(s,v,(1+\eps)^i).
    \end{equation}
    Moreover, observe that:
    \[d(e_{x,z})=(1+\eps)^{z+2(q-(k+1))}\leq (1+\eps)d(P_1)\cdot (1+\eps)^{2(q-(k+1))}=(1+\eps)^{2(q-k)-1}d(P_1).\]
    We conclude
    \begin{align*}
        d(e_{x,z}\cdot P_2)&=d(e_{x,z})+d(P_2)\\
        &\leq (1+\eps)^{2(q-k)-1}d(P_1)+d(P_2)\\
        &\leq (1+\eps)^{2(q-k)-1}d(P^*)\\
        &\leq (1+\eps)^{i+2(q-k)-1}.
    \end{align*}
    At the same time, the $(s,v)$-path $e_{x,z}\cdot P_2$ is contained in $G_s$ and consists of at most $2^{k+1}$ edges, so we have:
    \begin{equation}\label{eq:apsp2}
        \rdist^{2^{k+1}}_{G_s}\left(s,v,(1+\eps)^{i+2(q-k)-1}\right)\leq \rdist^{2^{k+1}}_{G_s}\left(s,v,d(e_{x,z}\cdot P_2)\right)\leq \ell(e_{x,z}\cdot P_2).
    \end{equation}
    By combining inequalities~\eqref{eq:apsp1}~and~\eqref{eq:apsp2}, we obtain
    \[\rdist^{2^{k+1}}_{G_s}\left(s,v,(1+\eps)^{i+2(q-k)-1}\right)\leq \dist(s,v,(1+\eps)^i),\]
    as desired.
\end{proof}
Observe that each of the $|V_{k}|$ used graphs $G_s$ has
\[O\left(m+|V_{k+1}|\log(n)/\eps\right)=O\left(m+(n/2^{k+1})\log^2(n)/\eps)\right)\]
edges.
Since the $\pidp{\one}$ is run with $h=2^{k+1}$, by~\Cref{thm:pi-dp}, the total time spent computing $\ap_k(\cdot,\cdot,\cdot)$ is
\[O\left(|V_k|\cdot 2^{k+1}\cdot \left(m+\frac{n}{2^{k+1}}\log^2(n)/\eps)\right)\log(n)/\eps\right)=O\left(\left(mn+\frac{n^2}{2^k}\log^2(n)/\eps\right)\log(n)/\eps\right).\]
Since there are $O(\log{n})$ indices $k$ to consider, and one needs to eventually run the algorithm with~$\eps$ decreased by a factor $\Theta(\log{n})$, the overall running time of the algorithm is:
\[O\left(\sum_{k=0}^q\left(mn+\frac{n^2}{2^k}\log^3(n)/\eps\right)\log^2(n)/\eps\right)=O\left(mn\log^3(n)/\eps+n^2\log^5(n)/\eps^2\right).\]
We have thus proved the following theorem.
\thmallpairs

\section{Open Problems}
\label{sec:open-problems}

We conclude our paper with a list of open problems.

\begin{itemize}
    \item Can we improve the runtime of $\Ot(mn^{3/5})$ for sparse graphs?
    The ideal would be a near/almost-linear time algorithm (as shown for undirected graphs by Bernstein \cite{Bernstein12}). 
    The main challenge here seems to be achieving a faster algorithm for DAGs.
    A fundamental obstacle to achieving almost-linear time is that our algorithm for sparse graphs is effectively based on a sparse approximate hopset, and such hopsets are known to be unable to reduce the (approximate) hop diameter to $n^{1/4-\Omega(1)}$ in directed graphs (see \cite{BodwinH23}).
    \item Our algorithm, as well as the almost linear-time algorithm for undirected graphs of \cite{Bernstein12}, suffer from a \emph{bicriteria} approximation: if the optimal path with delay at most $D$ has length~$L$, then we return a path with length at most $(1+\eps)L$ and delay at most $(1+\eps)D$.
        While the exact version of the problem is NP-hard \cite{GareyJ79}, there are several $O(mn)$ time algorithms that only incur a $(1+\eps)$-approximation on one of the parameters, while being exact on the other \cite{goel2001efficient,lorenz2001simple}.
        Even in undirected graphs, it remains open whether there exists an $o(mn)$ time algorithm that is exact on one of the parameters.
        There may be hope, since our $\Ot(n^2)$ and $\Ot(mn^{3/5})$ algorithms for DAGs incur an error in only one parameter.
        The error in the second parameter our main results suffer from originates from our framing of the problem as a reduction to Gap $\rsp$ and, for other reasons, from our using of LDDs by combining the lengths and delays into a single quantity.
    \item The $\Ot(mn)$ time algorithm of \cite{goel2001efficient} runs in strongly-polynomial time, and so do our algorithm for DAGs.
        On the other hand, our $o(mn)$-time results for general graphs have a dependence on the aspect ratio of lengths $W$ in the single-source variant; there is a $\log(nW)$ factor in the running times of \Cref{thm:dense-simple,thm:sparse-simple}.
        Can we remove this dependence on $W$?
        In fact, the almost-linear time algorithm of \cite{Bernstein12} also runs in strongly-polynomial time so there is not much evidence that suggests than an $o(mn)$ algorithm for directed graphs ought to have a dependence on $W$.
        The dependence on $W$, just like in the previous point, stems from our framing of the problem as a reduction to Gap $\rsp$.
        That being said, let us note that the dependence on $W$ can be avoided in the single source-target pair case (see~\Cref{lem:st-strong-poly}).
\end{itemize}

\bibliographystyle{alpha}
\bibliography{references}

\appendix
\section{Returning Paths}
\label{sec:returning-paths}

Our algorithms for $\rsp$ returned distances instead of actual paths.
In this section we briefly (and informally) sketch how to report an actual $(s,t)$-path $P$ using $\Ot(\card{P})$ additional time such that $\ell(P) \le (1+\eps)\dist(s,t,D)$ and $d(P) \le (1+\eps)D$.
We can view the algorithms of \Cref{thm:dense-simple,thm:sparse-simple} as a preprocessing step, and recover paths from their output in the way we next describe.

First, observe that we can easily modify $\pidp{\pi}$ so that paths can be recovered in time linear in their number of hops, by storing a back-pointer in each cell of $\dptable{\cdot}{\cdot}$.
We may then modify our reduction to Gap $\rsp$ by storing a pointer from every vertex to the test that witnesses its distance estimate.
The aforementioned test is comprised of an LDD Hierarchy, $\pi$, $H$, and the table $\dptable{\cdot}{\cdot}$ of $\pidp{\pi}$.
In this way we can recover, for any $t \in V$, the $(s,t)$-path $P$ corresponding to the distance answered in $\Ot(\card{P})$ time.
There is still one snag, however.
$P$ is a path in $G \cup H$ and we want to recover a path in $G$.
Accordingly, for every edge $uv \in P$ such that $uv \in H$, we will need a way to recover the $(u,v)$-path in $G$ that $uv$ is shortcutting.

If $cv \in H$ is a star-edge, we recover the $(c,v)$-path as follows.
As a preprocessing step, we modify the LDD Hierarchy construction by executing Dijkstra's algorithm on $\lendel{G}[C]$ and its reverse graph for all SCCs $C$, with the source chosen to be the representative vertex for $C$ (i.e the center of the star).
For each star-edge $cv$ in $C$, we add a pointer from $v$ to the corresponding entry in the Dijkstra table of $C$.
In this way, the $(c,v)$-path can be recovered using back-pointers within the Dijkstra table it points to.
This completes how the algorithm of \Cref{thm:dense-simple} can be used to report paths.

If $uv \in H$ is a hop-edge, we can use the All-Pairs $\rsp$ table which we used to create $uv$ to then recover the $(u,v)$-path.
Note however that our algorithm for All-Pairs $\rsp$ makes use of its own auxiliary edge set $H'$, and so the $(u,v)$-path may contain edges $e'$ from $H'$.
The shortcut edge $e'$, created from, say, the $k$th transition, can then be unfurled into a path containing edges from the $(k+1)$-st transition.
Repeating $O(\log n)$ many unfurlings, we get a path only using the edges of $G$, sketching out how the algorithm of \Cref{thm:sparse-simple} can be used to report paths.

\section{Parallel Edges}
\label{sec:parallel-edges}

We have presented our results under the assumption that the input graphs are simple so that the interesting aspects of our results would be highlighted.
In fact, our algorithms work for graphs that are not simple.
While \Cref{thm:sparse} goes through without any modification, we have to be slightly careful about the dense graph case (i.e., graphs with $\Omega(n^{3/2})$ edges) where \Cref{thm:dense} applies.

The part of our result there that is sensitive to parallel edges is the following bound
\begin{align*}
    \Pi
    = \sum_{e \in E} \frac{1}{\tau(e)} 
    = \sum_{u \in V} \sum_{\substack{v \in V: \\ uv \in E}} \frac{1}{\tau(v) - \tau(u)}
    \le \sum_{u \in V} \sum_{i \in [n]} \frac{1}{i}
    = O(n \log n)
\end{align*}
as seen in Example 2 of \Cref{sec:dp} and, more importantly, \Cref{obs:dense-Pi}.
The bound
\begin{align*}
    \sum_{u \in V} \sum_{\substack{v \in V: \\ uv \in E}} \frac{1}{\tau(v) - \tau(u)}
    \le \sum_{u \in V} \sum_{i \in [n]} \frac{1}{i}
\end{align*}
ceases to be true in the presence of parallel edges since $uv$, perhaps with $\tau(v) - \tau(u) = 1$, could occur up to $m$ times.
To rectify this situation, note that in giving a $(1 + \eps, 1 + \eps)$-approximation algorithm, we may assume that any pair of vertices $u,v \in V$ has at most $O(\log(n) / \eps)$ parallel edges between them.
Indeed, for each integer $k$ such that $(1 + \eps)^k$ is in the range $[\delta, D]$, we are forced to keep only the edge $uv$ with the smallest length such that $d(uv) \le (1 + \eps)^k$.
Since there are $O(\log_{1+\eps}(D/\delta))=O(\log n / \eps)$ such $k$, we can keep $O(\log(n) / \eps)$ parallel edges $uv$ and discard the rest. Then, we can bound $\Pi$ as follows:
\begin{align*}
    \Pi
    = \sum_{e \in E} \frac{1}{\tau(e)} 
    = \sum_{u \in V} \sum_{\substack{v \in V: \\ uv \in E}} \frac{O(\log(n)/\eps)}{\tau(v) - \tau(u)}
    \le \sum_{u \in V} \sum_{i \in [n]} \frac{O(\log(n)/\eps)}{i}
    = O(n \log^2(n)/\eps).
\end{align*}

In view of this, \Cref{thm:dense} can be amended to run in $O(m + n^2 \log^{10}(n) \log(nW) / \eps^4)$ time on multi-graphs.

\section{Reduction to Gap $\rsp$}
\label{sec:reduction-to-gap}

\reduction*
\begin{proof}
    Assume without loss of generality that $\dist(s,t,D) \ge 1$ for all $t$ (see \Cref{sec:zero-paths}).
    Consider the following reduction.
    \begin{tbox}
        \underline{\textbf{\textsc{Reduction to Gap $\rsp$}}}
        \begin{enumerate}
            \item $\eps^* \gets \frac{\eps}{4 c \log n}$, where $c$ is the constant hidden by the big-$O$s error guarantees of the assumed algorithm $\cA$ in the $\mathtt{NO}$ case.
            \item Let $L(t):=\infty$ for all $t\in V$.
            \item For $i$ such that $L_i = (1 + \eps^*)^i \in [1, (1+\eps^*)nW]$:
                \begin{enumerate}
                    \item Repeat $2 \log n$ times with fresh randomness each time:
                        \begin{enumerate}
                            \item Run $\cA$ on $G$ with source $s$, length threshold $L_i$, delay threshold $D$, and approximation $\eps^*$.
                            \item \label{alg:reduction-ii} For each $t \in V$, if $\cA$ said $\mathtt{YES}$ for $t$, then $L(t) \gets \min(L(t), (1 + O(\eps^* \log n))L_i)$.
                        \end{enumerate}
                \end{enumerate}
            \item Output $L(t)_{t\in V}$.
        \end{enumerate}
    \end{tbox}
    
    \underline{Running Time}:
    Clearly, the running time is $O\paren{(T(m,n) + n) \log^{\alpha + 2}(n) \log(nW) / \eps^{\alpha + 1}}$ since we have rescaled $\eps^* \gets \eps / (4c \log n)$ and $\cA$ is run $O(\log(nW)\log^2(n)\cdot \eps)$ times.

    \underline{Correctness}:
    We will show that the reduction gives a $((1 + O(\eps^* \log n))^2, (1 + O(\eps^* \log n))^2)$-approximate solution to $\rsp$.
    Since $(1 + O(\eps^* \log n))^2 = (1 + \eps/4)^2 \le e^{\eps / 2} \le 1 + \eps$, this would yield a $(1 + \eps, 1 + \eps)$-approximate solution to $\rsp$ as desired.
    
    Let $t \in V$ be an arbitrary vertex.
    \begin{itemize}
        \item If in iteration $i$ of the outer-loop $\cA$ ever answers $\mathtt{YES}$ for $t$, this means that there exists an $(s,t)$-path $P$ with $\ell(P) \le (1 + O(\eps^* \log n)) L_i$ and $d(P) \le (1 + O(\eps^* \log n))^2 D$.
            Crucially, the delay constraint is never violated by more than a $(1 + O(\eps^* \log n))^2$ factor.
        \item If $\dist(s,t,D)<\infty$, then $\dist(s,t,D)\leq nW$. In such a case, let $i$ be such that\linebreak $1 \le L_{i-1} \le  \dist(s,t,D) \le L_i \le (1+\eps^*)nW$.
            In iteration $i$ of the outer-loop, $\cA$ answers $\mathtt{YES}$ for $t$ with probability at least $1/2$.
            Since $\cA$ is called $2 \log n$ times, some call to $\cA$ will output $\mathtt{YES}$ with probability at least $(1 - 1/n^2)$ and thus
            $$L(t) \le (1 + O(\eps^* \log n)) L_i \le (1 + O(\eps^* \log n))\cdot (1+\eps^*) \dist(s,t,D)\leq (1 + O(\eps^* \log n))^2\dist(s,t,D).$$ If $\dist(s,t,D)=\infty$, then clearly $L(t)\leq \dist(s,t,D)$.
        \item Finally, note that if for a vertex $t$ the algorithm $\cA$ always answers $\mathtt{NO}$, then in particular there exists no $s\to t$ path with delay at most $(1 + O(\eps^* \log n))^2 D$. It follows that $\dist(s,t,D)=\infty$
        and thus returning $L(t)=\infty$ is correct.
    \end{itemize}
    By union bounding over all vertices, the pairs $(L(t),(1 + O(\eps^* \log n))^2 D)$ constitute \linebreak
    ${((1 + O(\eps^* \log n))^2, 1 + O(\eps^* \log n)^2)}$-approximate (i.e., $(1+\eps,1+\eps)$-approximate) solutions to the $\rsp$ problem for all $t \in V$ with probability at least $1 - 1/n$.
\end{proof}

\begin{lemma}\label{lem:st-strong-poly}
    Suppose our goal is to find an $(1+\eps,1+\eps)$-approximate solution to the $\rsp$ problem for a \emph{single} source-target pair $(s,t)$, rather than for all targets $t\in V$ at once, as required in~\Cref{prob:rsp}. Then the dependence on the aspect ratio in~\Cref{lem:reduction-to-gap} can be removed. 
\end{lemma}
\begin{proof}
Note that the dependence on $W$ comes from the fact that initially, we can only bound each $\rdist(s,t,D)$ by $1$ from the below, and by $nW$ from the above. However, if we were guaranteed that for all $t\in V$, $\rdist(s,t,D)\in [\gamma,n\cdot \gamma]$ for some number $\gamma$, then $\log(nW)$ factor incurred by the outer loop of the reduction would turn into a $\log((n\gamma)/\gamma)=\log(n)$ factor.

While it is not always the case that a single number $\gamma$ like this exists for all targets $t\in V$ in general, if one focuses on a single target $t$, such a $\gamma$ can be found in near-linear time. 
Indeed, define~$\gamma$ to be the smallest number such that the graph $G_\gamma=(V,E_\gamma)$, where
$E_\beta=\{e\in E:\ell(e)\leq \beta\}$, contains
an $(s,t)$-path with delay bounded by $D$.
Since every edge in $G_\gamma$ has length at most $\gamma$, we have $\rdist(s,t,D)\leq n\gamma$. On the other hand, every $(s,t)$-path in $G_\gamma$ with delay at most $D$ has to contain at least one edge with length at least $\gamma$, by the minimality of $\gamma$. This proves $\rdist(s,t,D)\geq \gamma$, as desired.

Note that for a candidate number $x$, one can test whether the sought value $\gamma$ satisfies $\gamma<x$ in $\Ot(m)$ time using Dijkstra's algorithm (wrt. delays) on $G_x$.
Moreover, observe that $\gamma$ is necessarily the length of some edge in $E$. As a result, $\gamma$ can be located via binary search on the set of edge lengths in $G$ in $\Ot(m)$ time.
\end{proof}

\section{Zero-Length Paths, and Normalizing Lengths}
\label{sec:zero-paths}

In this paper, we assume the lengths of edges in the input graph are in $\set{0} \cup [1,W]$ and, moreover, that the lengths of restricted shortest paths that our main results find are at least $1$.

For the first assumption, we can rescale all lengths \[\ell(e) \gets \frac{\ell(e)}{\min_{\substack{f \in E: \\ \ell(f) \neq 0}}(\ell(f))}.\]
This does not change the restricted shortest paths, since lengths of paths are sums over the edges and we have multiplied the length of each edge by the same scalar.

For the second assumption, we find all $t$ such that an $(s,t)$-path $P$ with $\ell(P) = 0$ and $d(P) \le D$ in $\Ot(m)$ time in the following way.

\begin{tbox}
    \begin{enumerate}
        \item $F \gets \set{e \in E : \ell(e) > 0}$
        \item Run Dijkstra's Algorithm on $G' = (V, E \setminus F, d)$
        \item Answer $0$ for each $t \in V$ if $\dist_{G'}(s,t) \le D$
    \end{enumerate}
\end{tbox}

\end{document}